\pdfoutput=1
\documentclass{article}
\usepackage{changepage}
\usepackage{stmaryrd}
\usepackage{mathtools}
\usepackage{mathpartir}
\usepackage{preamble}
\usepackage{caption}
\usepackage{subcaption}
\usepackage{iblock}
\usepackage{ebproof}
\usepackage{doi}

\title{Towards a geometry for syntax}
\author{Jonathan Sterling\thanks{University of Cambridge}}

\begin{document}

\maketitle

\begin{abstract}
It often happens that \emph{free} algebras for a given theory satisfy useful
reasoning principles that are not preserved under homomorphisms of algebras,
and hence need not hold in an arbitrary algebra. For instance, if $M$ is the
free monoid on a set $A$, then the scalar multiplication function $A \times M \rightarrow M$
is injective. Therefore, when reasoning in the \emph{formal theory of
monoids} under $A$, it is possible to use this injectivity law to make sound
deductions even about monoids under $A$ for which scalar multiplication is not
injective --- a principle known in algebra as the \emph{permanence of identity}.
Properties of this kind are of fundamental practical importance to the
logicians and computer scientists who design and implement computerized proof
assistants like Lean and Coq, as they enable the formal reductions of
equational problems that make type checking tractable.

As type theories have become increasingly more sophisticated, it has become
more and more difficult to establish the useful properties of their free models
that enable effective implementation. These obstructions have facilitated a
fruitful return to foundational work in type theory, which has taken on a more
geometrical flavor than ever before. Here we expose a modern way to prove a
highly non-trivial injectivity law for free models of Martin-L\"of type theory,
paying special attention to the ways that contemporary methods in type theory
have been influenced by three important ideas of the Grothendieck school: the
\emph{relative point of view}, the language of \emph{universes}, and the
\emph{recollement} of generalized spaces.
\end{abstract}

\nocite{grothendieck:rs}
\nocite{artin:2011}

\paragraph{Comment and acknowledgment}

This paper is an interpretation of the ideas of
\citet{awodey:2018:natural-models,bocquet-kaposi-sattler:2021,coquand:2019,fiore:2002,fiore:2012,newstead:2018,rijke-shulman-spitters:2020,shulman:blog:scones-logical-relations,shulman:2015,sterling:2021:thesis,uemura:2021:thesis,uemura:2022:coh}
as well as several other cited authors; the results described in this paper are
not new, but their explanation might be. In addition to the cited authors, I am
also greatly indebted to Mathieu Anel, Carlo Angiuli, Lars Birkedal, Daniel
Gratzer, and Robert Harper for years of enlightening conversations on these
topics. I thank also Chris Gossack for their helpful comments and suggestions.

This work was funded by the European Union under the Marie Sk\l{}odowska-Curie
Actions Postdoctoral Fellowship project
\href{https://cordis.europa.eu/project/id/101065303}{\emph{TypeSynth: synthetic
    methods in program verification}}. Views and opinions expressed are however
those of the authors only and do not necessarily reflect those of the European
Union or the European Commission. Neither the European Union nor the granting
authority can be held responsible for them.

\bigskip

\section{Introduction}

The purpose of this paper is to explain several ways in which the Grothendieck
school has influenced theoretical computer science, focusing on the
subdiscipline of \emph{type theory} and the study of its free models.

\subsection{Type theory and the \texorpdfstring{\DefEmph{relative point of view}}{relative point of view}}\label{sec:relative-pov}

Type theory is, of course, the study of Types; but much like other important
scientific and philosophical categories such as Space and Number, there is not
a single definition of what a Type is.  Although the field of type theory is
often said to have been born with Russell's
investigations~\citep{russell:1937,russell:1908} into a syntactic way to avoid the
eponymous ``paradox'', it must be said that type theory today has very little
in common with this early line of research. Type theory in the sense studied by
professionals is rather aimed to provide both informal and formal mathematical
language to speak of objects and structures varying ``continuously'' over a
base --- in other words, to define the mathematical foundations to
operationalize Grothendieck's \DefEmph{relative point of view}.

\begin{quote}
  The \DefEmph{relative point of view} states that instead of studying (\eg)
  schemes $X$ in the absolute, we should always study \DefEmph{relative} schemes
  $X\in \Kwd{Sch}/B$ for an arbitrary base $B$.
\end{quote}

Category theory implements the relative point of view by means of
\DefEmph{fibrations}; but this language is greatly obfuscated in
comparison to the simplicity of working with non-relative objects. The goal of
type theory is to reconcile the expressivity of the relative point of view with
the simplicity of the global point of view, by providing a language that makes
movement between different fibers (base change) seamless. Because type theory
is built up from very simple and abstract axiomatics, many categories of
interest possess \emph{type theoretic internal languages} which provide
streamlined accounts of relative objects~\citep{maietti:2005}.

\begin{example}[Relative schemes, type theoretically]
  We recall that a \emph{relative scheme} over a scheme $B$ is conventionally
  defined to be a morphism $\Mor{E}{B}$ in the category of schemes. In
  contrast, the type theoretic viewpoint turns the display of $E$ over $B$ on
  its side: in the type theoretic internal language of the (gros) Zariski
  topos~\citep{blechschmidt:2017}, a ``scheme'' $B$ is nothing more than a type
  satisfying certain conditions and a relative scheme over $B$ is nothing more
  than a scheme $E\prn{x}$ varying in a parameter $x:B$. The constraints of type
  theoretic language \emph{automatically} ensure that all the $E\prn{x}$ can be
  glued together to form a single scheme $\Sum{x:B}{E\prn{x}}$ and moreover
  that the projection $\Mor{\Sum{x:B}{E\prn{x}}}{B}$ is in fact a genuine
  morphism of schemes.  In this way, type theoretic language more directly
  captures the base intuitions of the relative point of view while minimizing
  bureaucratic overhead.
\end{example}

\NewDocumentCommand\ESP{}{\Kwd{Esp}}

\subsection{Universes in type theory and category theory}

Type theory replaces the display of relative objects $\Mor{E}{B}$ with
\emph{families} of objects $E\prn{x}$ varying in a formal parameter $x:B$.
This is achieved by postulating an imaginary ``object of all objects'' or a
\DefEmph{universal object} such that $B$-indexed families of objects can be
phrased in terms of functions \emph{into} the universal object. To start with,
this idea of postulating an imaginary universal object seems quite dangerous; for
instance, if types are interpreted as sets then this postulate seems to imply a
``set of all sets'' in which indexed-families of sets can be valued. More
concerningly, if types are interpreted as (\eg) topological spaces, to
postulate a universal object seems to imply a ``topological space of all
topological spaces'', which makes even less sense than a set of all sets.

It is a fundamental result of the field of type theory, however, that the
extension of a given theory by a universal object in this sense is
\DefEmph{conservative}.

\begin{theorem}[\citet{lumsdaine-warren:2015,awodey:2018:natural-models}]\label{thm:local-universes}
  Let $\CCat$ be a category with a class of morphisms
  $\mathcal{D}$ stable under pullback. Then the Yoneda embedding
  $\EmbMor[\Yo]{\CCat}{\Psh{\CCat}}$ of $\CCat$ into the (larger) category of
  presheaves on $\CCat$ has a \DefEmph{universal family}
  $\Mor[\Proj[\mathcal{D}]]{\EL[\mathcal{D}]}{\TP[\mathcal{D}]}$ such that
  every $\Mor[p]{E}{B}\in\mathcal{D}$ arises from it by pullback:
  \[
    \DiagramSquare{
      nw/style = pullback,
      nw = \Yo{E},
      sw = \Yo{B},
      ne = \EL[\mathcal{D}],
      se = \TP[\mathcal{D}],
      east = \Proj[\mathcal{D}],
      west = \Yo{p},
      south = \chi\Sub{p},
      north = \bar\chi\Sub{P},
      south/style = {exists,->},
      north/style = {exists,->},
    }
  \]

  Moreover, every fiber of $\EL[\mathcal{D}]$ over a representable is
  represented by an element of $\mathcal{D}$.
\end{theorem}

The import of the fundamental result above is that relative objects qua
morphisms $\Mor[p]{E}{B}$ in $\mathcal{D}$ can just as well be manipulated in
terms of their characteristic morphisms
$\Mor[\chi\Sub{p}]{\Yo{B}}{\TP[\mathcal{D}]}$ into the universal object.  It is
in this sense that type theory turns the display of relative objects ``on its
side''; note that the assumptions of \cref{thm:local-universes} are extremely
light and are easily accommodated in many scenarios of interest, as we see
below.

\begin{example}
  The following are examples of categories $\CCat$ equipped with a class of
  maps $\mathcal{D}$ satisfying the assumptions of \cref{thm:local-universes}:
  \begin{enumerate}

    \item The category $\SET$, $\ESP$, or $\Kwd{Sch}$ equipped with the
          class of all maps.

    \item The category $\ESP$ equipped with the class of fiber bundles, or of
          trivial bundles, \etc.

    \item The category of simplicial sets equipped with the class of Kan
          fibrations.

  \end{enumerate}
\end{example}

\subsubsection{Strict base change via universal objects}

The practical advantages of viewing relative objects in terms of morphisms into
a universal object can be articulated simply: whereas base change of
$\Mor[p]{E}{B}$ along $\Mor[b]{C}{B}$ must be implemented by pullback, the base
change of the characteristic map is given more simply by precomposition
$\Mor[\chi_p\circ \Yo{b}]{\Yo{C}}{\TP[\mathcal{D}]}$. The presentation in terms
of precomposition is simpler to work with because it is \emph{strictly
  associative and unital} in relation to base changes.

It is the strictness of base change qua precomposition that allows us to
directly speak of the fibers of a parameterized object $\prn{x:B}\mapsto
  E\prn{x}$, since for any $f : C \to B$ and $g : D\to C$ the notation
$E\prn{f\prn{gx}}$ cannot distinguish between ``first composing $f$ with $g$
and then doing base change'' and ``doing base change along $f$ and then base
change along $g$''. When base change is implemented by pullback, these two ways
to restrict $\Mor{E}{B}$ to $D$ are distinct but linked by a canonical
isomorphism. The strength of type theory is to completely avoid the need to
manipulate such canonical isomorphisms without sacrificing rigor.

\subsubsection{Grothendieck's universes}

As the terminology suggests, there is a great deal of similarity between the
idea of universal objects and Grothendieck's universes, which he famously
employed in SGA~4 to deal rigorously with the size issues that can quickly
arise when using category theory to organize mathematics~\citep{sga:4}. Indeed,
a given Grothendieck universe is a universal object for the class of maps in
$\SET$ whose fibers have cardinality strictly lower than a given strongly
inaccessible cardinal.

Grothendieck's universes were defined in terms of set theory and the
$\in$-relation, but subsequent developments by several authors (including
B\'enabou, Martin-L\"of, Hofmann, Streicher, and others) has led to a more
structural perspective on universe objects that is amenable to formalization in
an arbitrary category.
The most influential input has been that of Jean B\'enabou who had introduced
already in his 1971 lectures the notion of a \emph{universe in a
  topos}~\citep[\S6]{benabou:1973}, which is essentially a \emph{full internal
  subcategory} of the topos satisfying certain closure conditions, later
interpreted and developed substantially further by \citet{streicher:2005} and
several other authors.

\begin{remark}
  Grothendieck seems to have been unaware that \citet{zermelo:1930} had already considered both his notion of universe and his \emph{universe axiom}, as pointed out by \citet{hamkins:2022:authorship-of-grothendieck-universes}; for this reason, it may be most fair to refer to ``Zermelo--Grothendieck universes''. It seems most likely that Per Martin-L\"of's universes were conceived independently of Grothendieck's and with different motivation --- as one may deduce from the fact that the first version of Martin-L\"of's universe~\citep{martin-lof:1971} was plagued by the very antinomy that Grothendieck's universes aimed to avoid. Subsequent developments in the categorical--algebraic understanding of Martin-L\"of's type theoretic universes were, however, deeply influenced by the work of Grothendieck and B\'enabou.
\end{remark}

\subsubsection{Universes in a category}

\begin{definition}\label{def:universe}
  A \DefEmph{universe} $\SS$ in a category $\CCat$ is given by a single
  carrable morphism $\Mor[\Proj[\SS]]{\EL[\SS]}{\TP[\SS]}$ called its
  \DefEmph{generic family}.\footnote{A carrable morphism is one along which all
    pullbacks exist.} For a morphism $\Mor[f]{E}{B}$ in $\CCat$, we will write
  $f\in\SS$ or ``$f$ is \DefEmph{classified by} $\SS$'' to mean that $f$ arises
  by pullback from $\Mor[\Proj[\SS]]{\EL[\SS]}{\TP[\SS]}$, \ie there exists a
  cartesian map $\Mor{f}{\Proj[\SS]}$ in the fundamental (codomain) fibration
  $\FFib{\CCat}$.
\end{definition}

\begin{remark}
  A universe $\SS$ in a locally cartesian closed category determines an internal
  category in $\CCat$, whose object of objects is $\TP[\SS]$ itself and whose
  object of morphisms is the exponential
  $\pi_1^*\EL[\SS]\Rightarrow\pi_2^*\EL[\SS]$ over $\TP[\SS]\times\TP[\SS]$.
  The externalization of the internal category $\SS$ is then a full
  subfibration of the fundamental fibration $\FFib{\CCat} = \prn{\FibMor[\Con{cod}]{\CCat^\to}{\CCat}}$. For each $I\in
    \CCat$, the fiber of this full subfibration is given by morphisms
  $\Mor[e]{E}{I}$ that are \DefEmph{classified} by $\SS$. In other words
  $\Mor[e]{E}{I}$ lies in the full subfibration when there exists a cartesian
  morphism $\Mor{e}{\Proj[\SS]}$.
\end{remark}

\begin{definition}[Contextual class of objects]\label{def:contextual-class}
  Let $\SS$ be a universe in a category $\CCat$ with a terminal object; a class of
  objects $\mathcal{X}\subseteq\CCat$ is called \DefEmph{$\SS$-contextual} when it
  satisfies the following closure conditions:
  \begin{enumerate}
    \item the terminal object is contained in $\mathcal{X}$;
    \item if $C\in\CCat$ is contained in $\mathcal{X}$ and $\Mor[p]{A}{C}$ lies in $\SS$, then $A$ is in $\mathcal{X}$.
  \end{enumerate}
\end{definition}

\begin{definition}[Contextual objects]\label{def:contextual-objects}
  Let $\SS$ be a universe in a category $\CCat$ with a terminal object; an object
  of $\CCat$ is called \DefEmph{$\SS$-contextual} when it is contained in the
  smallest $\SS$-contextual class in the sense of \cref{def:contextual-class}.
\end{definition}

\subsubsection{Grothendieck--B\'enabou universes inside a topos}

If $\ECat$ is an elementary topos with a natural numbers object $N$, following
\citet{benabou:1973,streicher:2005} we can define a notion of
\DefEmph{Grothendieck--B\'enabou universe} in $\ECat$ that restricts to the
familiar notion of Grothendieck universe when $\ECat=\SET$.

\begin{definition}\label{def:gb-universe}
  A universe $\SS$ in $\ECat$ is called a \DefEmph{Grothendieck--B\'enabou
    universe} when it satisfies the following conditions:
  \begin{enumerate}

    \item \DefEmph{dependent sums} and \DefEmph{dependent products}: if both
          $\Mor[e]{E}{B}$ and $\Mor[b]{B}{C}$ are classified by $\SS$, then both
          $b_!e$ and $b_*e$ are classified by $\SS$, where $b_!\dashv b^*\dashv
            b_*$ is the base change adjoint triple.
    \item \DefEmph{propositional resizing}: every monomorphism of $\ECat$ is
          classified by $\SS$.
    \item \DefEmph{descent}: for any $g$ and $f$ such that
          $g$ is classified by $\SS$, if there is a cartesian epimorphism
          $\Mor{g}{f}$ in $\FFib{\ECat}$, then $f$ is classified by $\SS$.
    \item \DefEmph{subobject classifier}: $\Mor{\Omega}{\ObjTerm{\ECat}}$ is
          classified by $\SS$.
    \item \DefEmph{natural numbers object}: $\Mor{N}{\ObjTerm{\ECat}}$ is
          classified by $\SS$.
  \end{enumerate}
\end{definition}

\subsection{Abstract and concrete syntax of type theory}\label{sec:abstract-syntax}

So far we have discussed type theory as a convenient notation for working with
relative objects in various categories. Most users of type theory will need \emph{no more} than this informal perspective on type theory. In
order to more thoroughly justify these applications, however, type theorists
have rendered the interpretation of type theoretical notations in various
categories as part of a more general discourse on the \emph{syntax and
  semantics of type theory}~\citep{hofmann:1997}.

There are many ways to think about what a model of type theory ought to be, but
most of them take the form of categories $\CCat$ equipped with additional
structure in $\Psh{\CCat}$, axiomatizing the scenario of
\cref{thm:local-universes}. The syntax of type theory can be studied both
abstractly and concretely; the \DefEmph{concrete syntax} of type theory can be
defined in terms of a (very complex) formal grammar, but it is just as well to
define the \DefEmph{abstract syntax} of a type theory to be given by the
initial object in the category of models of that type theory. That abstract
syntax can in fact be constructed as a quotient of concrete syntax is a
consequence of the results of \citet{cartmell:1978}, later tackled in more
specificity by \citet{streicher:1991} and \citet{uemura:2021:thesis}.  Renewed
interest during the past decade~\citep{voevodsky:2016:templeton} has led to
several creative re-treadings of the ground first paved by
\citeauthor{cartmell:1978}.

\subsubsection{Computerized proof assistants}

One motivation to study the syntax of type theory is to facilitate its
\emph{implemention} in computerized proof assistants; these are tools into
which human beings can enter formal type theoretical expressions representing
mathematical objects and proofs and have their validity automatically checked.
In addition to assuring the validity of constructions and proofs, proof
assistants also assist with book-keeping tasks --- such as displaying what it
remains to show at any given point in an incomplete proof.
Dependent type theoretic proof assistants such as Coq~\citep{coq:reference-manual},
Lean~\citep{lean:2015}, and Agda~\citep{norell:2009} are now routinely used to
develop and verify the correctness of both old and new
mathematics~\citep{gonthier:2008,gonthier-et-al:2013,favonia-finster-licata-lumsdaine:2016,scholze:2022:liquid-tensor},
and there are now very extensive and mature libraries of mathematical results
available~\citep{mathlib:2020,mahboubi-tassi:2020,unimath,agda-unimath,type-topology,1lab}.

\subsubsection{External \texorpdfstring{\vs}{vs.} internal equality}

Type theory is a somewhat unique language, in that it contains two different
kinds of equality: external and internal. Type theory's \DefEmph{external
  equality} is simply the congruence under which assertions of the form $u:A$ are
stable; in particular, when a type $A$ is \emph{externally equal} to a type
$B$, written $A\equiv B$, we may assert $u:A$ if and only if we may assert
$u:B$. True to its name, external equality cannot be assumed or refuted inside
type theory; in other words, it is part of the \emph{grammar} rather than the
\emph{vocabulary} of type theory.\footnote{External equality in our sense is
  usually referred to as \DefEmph{judgmental equality} or \DefEmph{definitional
    equality}; both the traditional terminologies carry some philosophical force
  and subtlety that we do not necessarily intend, so we prefer our more neutral
  terminology. We refer the reader to
  \citet{martin-lof:1975,martin-lof:1996,martin-lof:1987} for further discussion
  of the philosophical aspects.} In a \emph{model of type theory} (including the
initial model), external equality is interpreted as ordinary mathematical
equality between elements of the model.

The second kind of equality in type theory is \DefEmph{internal equality},
which is part of the vocabulary of type theory. For every type $A$ and elements
$u,v:A$ there is a third type $u=_Av$ classifying \emph{identifications} of
$u$ and $v$ as elements of $A$. Internal equality is meant to correspond to
ordinary mathematical equality; so, for instance, if the notion of a
\emph{group} is formalized in type theory, the unit laws are stated in terms of
internal equality. Here we shall not dwell further on internal
equality, in spite of the fact that it has been the main topic of type
theoretic research for more than two
decades~\citep{hofmann-streicher:1998,voevodsky:2006,awodey-warren:2009,hottbook}.

\subsubsection{Decidability of external equality}\label{sec:decidability}

Although there are a variety of possible designs for computerized proof
assistants based on type theory, experience has verified that the most practical
approach is to ensure that the relation of external equality can be
\emph{automatically} checked by the computer without any intervention by the
user. This goal, however, places severe constraints on what kinds of equations
can be part of external equality --- as it is easy for external equality to become
\emph{undecidable} if enough laws are
added~\citep{castellan-clairambault-dybjer:2017}. For this reason, type theorists
have accumulated a variety of design principles that tend to ensure effective
decidability --- though it remains very difficult to establish decidability in
any specific case.

\subsubsection{Running example: injectivity of type constructors}\label{sec:injectivity}

In addition to decidability, one of the key lemmas servicing the computerized
implementation of type theory is the \DefEmph{injectivity of type
  constructors}, which is what allows an algorithm to universally decompose the
task of checking an equation like $A\Rightarrow B \equiv A' \Rightarrow B'$ to
the task of checking both $A\equiv A'$ and $B\equiv B'$: the injectivity
property states that the latter judgments are the \emph{only} way that the two
function spaces could be equal. Note that injectivity in this sense does
\emph{not} imply that the $\prn{\Rightarrow}$ operator is a monomorphism in an
arbitrary model of type theory (indeed, doing so would rule out most semantic
models of interest!). Nonetheless, the injectivity property can be stated in
terms of $\prn{\Rightarrow}$ being a monomorphism in the \emph{initial model}
of type theory. In fact, we shall use this injectivity law as our running
example throughout the rest of this paper.

\subsection{Normalization and injectivity, for free monoids}\label{sec:monoid-normalization}

Type theorists have found that the most reliable way to establish both
decidability of external equality~(\cref{sec:decidability}) and injectivity of
type constructors~(\cref{sec:injectivity}) is to devise a \emph{concrete}
characterization of equivalence classes of expressions in terms of
\DefEmph{normal forms}, equipping the quotient of concrete syntax by external
equality with a more canonical section that is amenable to effective
computation. This process is referred to as \DefEmph{normalization}.

Normalization is better understood first in a simpler context; to that end, we
consider the theory of monoids below and a similar injectivity law that we
might wish to establish for \emph{free} monoids.

\subsubsection{The theory of monoids}

The algebraic theory of monoids
subjects a nullary operation $\epsilon$ and a binary operation $\mu$ to
the following equational laws:
\begin{align*}
  x     & \vdash \mu\prn{\epsilon,x} \equiv x                           \\
  x     & \vdash \mu\prn{x,\epsilon} \equiv x                           \\
  x,y,z & \vdash \mu\prn{\mu\prn{x,y},z} \equiv \mu\prn{x,\mu\prn{y,z}}
\end{align*}

\NewDocumentCommand\Expr{m}{\Con{Expr}\,A}
\NewDocumentCommand\FreeMonoid{}{\Con{F}}

\subsubsection{Constructing the free monoid on a set}

Given a set $A$, we may construct the \emph{free monoid} on $A$ by
taking a quotient of the well-formed expressions in the theory of monoids with
$\vrt{A}$-many additional constants. First we may inductively define the set
$\Expr{A}$ of expressions by the generators:
\begin{mathpar}
  \ebrule{
    \hypo{
      a\in A
    }
    \infer1{
      \eta\prn{a} \in\Expr{A}
    }
  }
  \and
  \ebrule{
    \infer0{
      \epsilon\in\Expr{A}
    }
  }
  \and
  \ebrule{
    \hypo{u\in\Expr{A}}
    \hypo{v\in\Expr{A}}
    \infer2{
      \mu\prn{u,v}\in\Expr{A}
    }
  }
\end{mathpar}

Next we inductively define
$\prn{\sim}\subseteq\Expr{A}\times\Expr{A}$ to be the smallest congruence for the operations above closed under the following rules:
\begin{mathpar}
  \ebrule{
    \hypo{
      u\in\Expr{A}
    }
    \infer1{
      \mu\prn{\epsilon,u} \sim u
    }
  }
  \and
  \ebrule{
    \hypo{
      u\in\Expr{A}
    }
    \infer1{
      \mu\prn{u,\epsilon}\sim u
    }
  }
  \and
  \ebrule{
    \hypo{u\in \Expr{A}}
    \hypo{v\in \Expr{A}}
    \hypo{w\in \Expr{A}}
    \infer3{
      \mu\prn{\mu\prn{u,v},w} \sim \mu\prn{u,\mu\prn{v,w}}
    }
  }
\end{mathpar}

Then the carrier set of the free monoid on $A$ can be expressed as the quotient
$\FreeMonoid{A} = \Expr{A}/{\sim}$. Because $\prn{\sim}$ is a congruence, there
is an evident monoid structure on $\FreeMonoid{A}$ and it is simple to show that
this monoid structure is universal in relation to monoids equipped whose
carriers lie underneath $A$.

\subsubsection{Injectivity of scalar multiplication in the free monoid}\label{sec:free-monoid-nf}

There is a ``scalar multiplication'' $A\times \FreeMonoid{A}\to \FreeMonoid{A}$
function on the free monoid sending $\prn{a,u}$ to $\mu\prn{\eta\prn{a},u}$.
A monoid-theoretic analogue to our running example (\cref{sec:injectivity})
would be to prove that the scalar multiplication function on free monoids is
injective. With our presentation of $\FreeMonoid{A}$ as a quotient of $\Expr{A}$,
however, it is very hard to see that this is necessariy the case --- as we do
not have any kind of a handle on equivalence classes.

The solution is to find an alternative presentation of the free monoid that can
be defined inductively without any quotienting; and such an alternative
presentation is referred to a \DefEmph{normal forms presentation}. In the case
of free monoids, there is a trivial candidate for the normal forms
presentation: the set of \emph{lists} $A^\star$ of elements of $A$, which can
be defined inductively as follows:
\begin{mathpar}
  \ebrule{
    \infer0{
      \brk{}\in A^\star
    }
  }
  \and
  \ebrule{
    \hypo{a\in A}
    \hypo{m\in A^\star}
    \infer2{
      a\lhd m \in A^\star
    }
  }
\end{mathpar}

In other words, $A^\star$ is the initial algebra for the polynomial endofunctor
$F\prn{X} = \ObjTerm{} + A\times X$. There is no need to quotient $A^\star$;
the monoid operations are defined by the following equations, using the
induction principle of $A^\star$:
\begin{align*}
  \eta\Sub{A^\star}\prn{a}        & = a\lhd\brk{}
  \\
  \epsilon\Sub{A^\star}           & = \brk{}
  \\
  \mu\Sub{A^\star}\prn{\brk{},n}  & = n
  \\
  \mu\Sub{A^\star}\prn{a\lhd m,n} & = a\lhd \mu\Sub{A^\star}\prn{m,n}
\end{align*}

It is easy to show that $A^\star$ satisfies the equational laws of the monoid
theory, again by induction on lists. But more importantly, it is possible to
deduce immediately that the scalar multiplication on $A^\star$ is injective.

\begin{theorem}\label{thm:cons-injective}
  The scalar multiplication function $A\times A^\star\to A^\star$ sending each
  $\prn{a,m}$ to $\mu\Sub{A^\star}\prn{\eta\Sub{A^\star},m}$ is injective.
\end{theorem}

\begin{proof}
  Unfolding definitions, the scalar multiplication function is exactly the
  $\prn{\lhd}$ operation on lists; writing $\Mor[\alpha]{\ObjTerm{}+A\times
      A^\star}{A^\star}$ for structure map of $A^\star$ as an initial $F$-algebra,
  we recall that $\alpha$ is an isomorphism by Lambek's lemma, and so the
  constructor $\prn{\lhd}$ is the right coproduct inclusion, which is injective
  as coproducts of sets are disjoint.
\end{proof}

Therefore to deduce from \cref{thm:cons-injective} that the scalar
multiplication function on $\FreeMonoid{A}$ is injective, it suffices to construct
an isomorphism of monoids under $A$ between $\FreeMonoid{A}$ and $A^\star$; in fact,
it is even enough to exhibit $\FreeMonoid{A}$ as a retract of $A^\star$, as depicted
below where the horizontal arrow is the unique homomorphism of monoids under
$A$ determined by the universal property of $F\Sub{\mathbb{T}}$:
\[
  \begin{tikzpicture}[diagram]
    \node (nw) {$\FreeMonoid{A}$};
    \node[right = of nw] (ne) {$A^\star$};
    \node[below = of ne] (se) {$\FreeMonoid{A}$};
    \draw[>->] (nw) to node[above] {$S$} (ne);
    \draw[double] (nw) to (se);
    \draw[->>,exists] (ne) to node[right] {$P$} (se);
  \end{tikzpicture}
\]

\begin{corollary}\label{cor:scalar-multiplication-injective}
  The scalar multiplication function on the free monoid $\FreeMonoid{A}$ is injective.
\end{corollary}

\begin{proof}
  We define a retraction $\Mor[P]{A^\star}{\FreeMonoid{A}}$ of the universal
  map $S$, setting $P\brk{} = \epsilon$ and $P\prn{a\lhd m}$ to
  $\mu\prn{\eta\prn{a},P\prn{m}}$.
  Now fix $a,a'\in A$ and $u,u'\in \FreeMonoid{A}$ such that
  $\mu\prn{\eta\prn{a},u} = \mu\prn{\eta\prn{a'},u'}$. Applying the section $S$
  and using the fact that it is a homomorphism, we have $a\lhd S u = a'\lhd S
    u'$; by \cref{thm:cons-injective} we have $a=a'$ and $Su = Su'$. From the
  latter we deduce $P\prn{Su} = P\prn{S'u}$; because $P$ is a retraction of
  $S$, it follows that $u=u'$.
\end{proof}

\section{Free models of type theory and normalization}\label{sec:normalization}

The normalization problem for free monoids that we explored in
\cref{sec:monoid-normalization} is a particularly easy case. Unfortunately,
things become significantly more difficult when we move from simple algebraic
theories to full type theories, where we are trying to characterize the
equivalence classes of \emph{types} by normal forms; the difficulty is roughly
that types and their normal forms do not (a priori) live in the same category,
in contrast to the situation with monoids where both elements and normal forms
are organized into sets.

\subsection{Natural models of type theory}

We have alluded in \cref{sec:abstract-syntax} to the many notions of ``model of
type theory''; here we will consider \DefEmph{natural
  models}~\citep{awodey:2018:natural-models}, a categorical reformulation of
Dybjer's \emph{categories with families}~\citep{dybjer:1996}.

\subsubsection{Representable maps and natural models}

The definition of a natural model involves the concept of
\DefEmph{representable natural transformation}, which was incidentally
introduced by \citet{ega:1} in EGA~1.

\begin{definition}[Relative representability]
  Let $\CCat$ be a full subcategory of $\ECat$; a morphism $\Mor{E}{B}$ of
  $\ECat$ is said to be \DefEmph{relatively representable by an object of
    $\CCat$} when for any $\Mor{\Gamma}{B}$ such that $\Gamma$ lies in $\CCat$,
  the fiber product $E\times_B\Gamma$ lies in $\CCat$, identifying $\CCat$ with
  the essential image of the Yoneda embedding $\EmbMor[\Yo]{\CCat}{\Psh{\CCat}}$.
\end{definition}

\begin{definition}[\citet{awodey:2018:natural-models}]\label{def:nat-mod}
  A \DefEmph{natural model} $\Mod$ is defined to be an essentially small
  category $\CX[\Mod]$ with a terminal object and a natural transformation
  $\Mor[\Proj[\Mod]]{\EL[\Mod]}{\TP[\Mod]}$ in $\Psh{\CX[\Mod]}$ that is
  relatively representable by an object of $\CCat$.\footnote{Some previous
    expositions required $\CX[\Mod]$ to be small; nonetheless, the theory
    develops much more smoothly if we only require essential smallness.}
\end{definition}

Observe that a natural model is really a special kind of \emph{universe}
(\cref{def:universe}) in a category of presheaves.

\begin{exegesis}
  In a natural model $\Mod$, objects of $\CX[\Mod]$ are referred to as
  \emph{contexts} and morphisms are called \emph{substitutions}. When
  $\Gamma\in\CX[\Mod]$ is a context, an element of $\TP[\Mod]\Gamma$ is a
  \emph{type} $A\prn{\gamma}$ that depends on a parameter $\gamma:\Gamma$; the
  representability of $\Proj[\Mod]$ ensures for each type $A$ over $\Gamma$,
  there is a context $\Gamma.A$ called the \DefEmph{context comprehension}
  classifying pairs $\prn{\gamma,a}$ where $\gamma:\Gamma$ and $a$ is an
  element of $A\gamma$.
\end{exegesis}

\begin{definition}[Democratic natural models]\label{def:democratic}
  A natural model $\Mod$ is called \DefEmph{democratic} when every object
  $\Gamma\in\CX[\Mod]$ represents a $\Mod$-contextual object in the sense of
  \cref{def:contextual-objects}.
\end{definition}

\subsubsection{Function spaces on a natural model}

Further structures on a natural model can be imposed; for instance, function
spaces correspond to cartesian squares of the following form in $\JDG[\Mod]$:
\[
  \DiagramSquare{
    nw/style = pullback,
    width = 5.5cm,
    north/style = {exists,->},
    south/style = {exists,->},
    nw = \Sum{A,B:\TP[\Mod]}\Proj[\Mod]\Sup{-1}A\Rightarrow \Proj[\Mod]\Sup{-1}B,
    sw = \TP[\Mod]\times\TP[\Mod],
    ne = \EL[\Mod],
    se = \TP[\Mod],
    east = \Proj[\Mod],
    west = \prn{A,B,f}\mapsto\prn{A,B},
    south = \prn{\Rightarrow\Sub{\Mod}},
    north = \prn{\lambda\Sub{\Mod}},
  }
\]

It can be shown that a natural model can be equipped with function spaces if
and only if the corresponding universe is closed under pushforwards of product
projection maps (this is a restriction of the condition that a
Grothendieck--B\'enabou universe be closed under dependent products).

\subsubsection{The (2,1)-category of natural models}

Natural models and their structured variants (\eg natural models with function
spaces, \etc) all arrange into (2,1)-categories. Here we will not dwell on the
conditions for a morphism between natural models to extend to a morphism of
natural models with function spaces; the interested reader should consult
\citet{uemura:2021:thesis} for more on this.

\begin{definition}[\citet{newstead:2018}]
  Let $\Mod$ and
  $\Modd$ be two natural
  models.  A \DefEmph{pre-morphism} of natural models
  $\Mor[F]{\Mod}{\Modd}$ is given by a functor
  $\Mor[F]{\CX[\Mod]}{\CX[\Modd]}$ preserving the terminal object together with a square
  $\Mor[\Proj[F]]{F_!\Proj[\Mod]}{\Proj[\Modd]}$ in $\JDG[\Modd]$:
  \[
    \begin{tikzpicture}[diagram]
      \SpliceDiagramSquare{
        north/style = {->,exists},
        south/style = {->,exists},
        north = F\Sub{\EL},
        south = F\Sub{\TP},
        nw = F_!\EL[\Mod],
        sw = F_!\TP[\Mod],
        ne = \EL[\Modd],
        se = \TP[\Modd],
        west = F_!\Proj[\Mod],
        east = \Proj[\Modd],
      }
      \node [between = nw and se] {$\Proj[F]$};
    \end{tikzpicture}
  \]
\end{definition}

\begin{notation}
  Given a pre-morphism $\Mor[F]{\Mod}{\Modd}$ and a type
  $\Mor{\Yo{\Gamma}}{\TP[\Mod]}$, we shall write $\Mor[F\Sub{\TP}\cdot
      A]{F_!\Yo{\Gamma}}{\TP[\Modd]}$ for the composite $F\Sub{\TP}\circ F_!A$; we
  impose a similar notation on elements, setting $F\Sub{\EL}\cdot a$ to be
  $F\Sub{\EL}\circ F_!a$.
\end{notation}

\begin{definition}[\citet{newstead:2018}]
  A pre-morphism $\Mor[F]{\Mod}{\Modd}$ is said to be a
  \DefEmph{morphism of natural models} if it \DefEmph{preserves context
    comprehensions} in the sense that for every $\Gamma\in\CX[\Mod]$ and
  $\Mor[A]{\Yo{\Gamma}}{\TP[\Mod]}$, the composite square below is
  cartesian~\citep{newstead:2018}:
  \[
    \begin{tikzpicture}[diagram]
      \SpliceDiagramSquare<l/>{
        nw = F_!\Yo\prn{\Gamma.A},
        sw = F_!\Yo{\Gamma},
        west = F_!\Yo{p_A},
        ne = F_!\EL[\Mod],
        se = F_!\TP[\Mod],
        east = F_!\Proj[\Mod],
        east/node/style = upright desc,
        south = F_!A,
        north = F_!\Con{x}_A,
        width = 2.5cm,
      }
      \SpliceDiagramSquare<r/>{
        glue = west, glue target = l/,
        ne = \EL[\mathrlap{\Modd}],
        se = \TP[\mathrlap{\Modd}],
        east = \Proj[\Mod],
        south = F\Sub{\TP},
        north = F\Sub{\EL},
      }
    \end{tikzpicture}
  \]
\end{definition}

\begin{definition}[\cite{uemura:2021:thesis}]
  Let $\Mor[F,G]{\Mod}{\Modd}$ be two morphisms of natural
  models. An \DefEmph{isomorphism} from $F$ to $G$ is defined to be
  a natural isomorphism $\alpha : F \cong G$ between the underlying functors
  such that for each $\Mor[A]{\Yo{\Gamma}}{\TP[\Mod]}$, the black
  triangles below commute:
  \begin{mathparpagebreakable}
    \begin{tikzpicture}[diagram]
      \node (nw) {$F_!\Yo{\Gamma}$};
      \node (sw) [below = of nw] {$G_!\Yo{\Gamma}$};
      \node (ne) [right = 3cm of nw] {$\TP[\Modd]$};
      \draw[->] (nw) to node [above] {$F\Sub{\TP}\cdot A$} (ne);
      \draw[->] (nw) to (sw);
      \draw[->] (sw) to node [sloped,below] {$G\Sub{\TP}\cdot A$} (ne);
      \node[gray,left = of nw] (nww) {$\Yo{F\Gamma}$};
      \node[gray,left = of sw] (sww) {$\Yo{G\Gamma}$};
      \draw[->,gray] (nww) to node [left] {$\Yo{\alpha_\Gamma}$} (sww);
      \draw[->,gray] (nw) to node [above] {$\cong$} (nww);
      \draw[->,gray] (sww) to node [below] {$\cong$} (sw);
    \end{tikzpicture}
    \\
    \begin{tikzpicture}[diagram]
      \node (nw) {$F_!\prn{\Gamma.A}$};
      \node (sw) [below = of nw] {$G_!\prn{\Gamma.A}$};
      \node (ne) [right = 3cm of nw] {$\EL[\Modd]$};
      \draw[->] (nw) to node [above] {$F\Sub{\EL}\cdot \Con{x}_A$} (ne);
      \draw[->] (nw) to (sw);
      \draw[->] (sw) to node [sloped,below] {$G\Sub{\EL}\cdot \Con{x}_A$} (ne);
      \node[gray,left = 2.5cm of nw] (nww) {$\Yo{F\prn{\Gamma.A}}$};
      \node[gray,left = 2.5cm of sw] (sww) {$\Yo{G\prn{\Gamma.A}}$};
      \draw[->,gray] (nww) to node [left] {$\Yo{\alpha\Sub{\Gamma.A}}$} (sww);
      \draw[->,gray] (nw) to node [above] {$\cong$} (nww);
      \draw[->,gray] (sww) to node [below] {$\cong$} (sw);
    \end{tikzpicture}
  \end{mathparpagebreakable}
\end{definition}

\subsubsection{Free natural models: the abstract syntax of type theory}

What is important for us is that the (2,1)-category of natural models (and its
structured variants) be compactly generated or \emph{presentable} in the sense
of \citet{lurie:2009} and therefore have \emph{free objects}, \ie the
bi-initial natural model with function spaces on some constants, \etc. Note
that the bi-initial natural model of a given type theory is always democratic
in the sense of \cref{def:democratic}.

\subsubsection{From universes to natural
  models}\label{sec:universe-to-natural-model}

Let $\ECat$ be a locally small locally cartesian closed category, and let $\SS$
be a universe in $\ECat$ such that $\ObjTerm{\ECat}\in\SS$. Furthermore let
$\CCat\subseteq\ECat$ be a full subcategory closed under all contextual objects
with respect to $\SS$ in the sense of \cref{def:contextual-objects}. In this
section, we will define a natural model $\Extern[\CCat]{\SS}$ called the
\emph{externalization} of $\SS$ over $\CCat$. We define
$\CX[\Extern[\CCat]{\SS}]$ to be $\CCat$ itself. The inclusion functor
$\EmbMor[I]{\CCat}{\ECat}$ inducing a \emph{nerve}
$\Mor[N\Sub{\CCat}]{\ECat}{\Psh{\CCat}}$ sending each object to its
\DefEmph{functor of $\CCat$-valued points}:
\begin{align*}
  N\Sub{\CCat} & : \Mor{\ECat}{\Psh{\CCat}} \\
  N\Sub{\CCat} & : \Mor|{|->}|{E}{
  \Hom{\ECat}{I-}{E}
  }
\end{align*}

\begin{lemma}\label{lem:nv-preserves-representable-maps}
  Let $\Mor[p]{E}{B}$ be a morphism in $\ECat$ that is relatively representable
  by an object of $\CCat$; then $N\Sub{\CCat}\prn{\Mor[p]{E}{B}}$ is a
  representable natural transformation in $\Psh{\CCat}$.
\end{lemma}

\begin{proof}
  We must check that for any $\Gamma\in\CCat$, the fiber product of any cospan
  $\Yo{\Gamma}\xrightarrow{A} N\Sub{\CCat} B \xleftarrow{N\Sub{\CCat} p} N\Sub{\CCat} E$ is
  representable. Identifying $\Yo{\Gamma}$ with $N\Sub{\CCat} I\Gamma$, we may assume
  without loss of generality that $A= N\Sub{\CCat} A'$ for some
  $\Mor{A'}{I\Gamma}{B}$. The fiber product of
  $I\Gamma\xrightarrow{A'}B\xleftarrow{p}E$ is \emph{contextual} with respect to
  $\SS$ by definition, and it lies in $\CCat$ by our assumption that $\CCat$
  contains all contextual objects.
\end{proof}

We may therefore define the generic family $\Proj[\Extern[\CCat]{\SS}]$ of
$\Extern[\CCat]{\SS}$ to be $N\Sub{\CCat}\Proj[\SS]$, which is representable by
\cref{lem:nv-preserves-representable-maps}.

\subsection{Injectivity of type constructors in free natural models}

We now come to a precise version of our original discussion about
injectivity of type constructors from \cref{sec:injectivity}.

\begin{question}\label{question:injectivity}
  Let $\IMod$ be the free natural model with function spaces generated by a
  base type $\Con{O}$ and two constants $\Con{yes},\Con{no}:\Con{O}$. Is the
  function space constructor
  $\Mor[\prn{\Rightarrow\Sub{\IMod}}]{\TP[\IMod]\times\TP[\IMod]}{\TP[\IMod]}$
  a monomorphism in $\JDG[\IMod]$?
\end{question}

The answer to \cref{question:injectivity} is ultimately ``Yes!''
(\cref{thm:main-theorem}), but this is as difficult to prove as
\cref{cor:scalar-multiplication-injective} was easy. As we have alluded to at
the beginning of \cref{sec:normalization}, the problem is that
although the collection of types is a presheaf on $\CX[\IMod]$, we cannot very
well define the collection of \emph{normal forms of types} to be a presheaf on
$\CX[\IMod]$, as we will illustrate in \cref{sec:nf-not-presheaf}.

\subsection{Normal forms are not functorial in substitutions}\label{sec:nf-not-presheaf}

The reason that a useful notion of normal form cannot be defined as a presheaf
on $\CX[\IMod]$ is that normal forms must distinguish between
\DefEmph{variables} and the things that can be substituted for them. For
instance, if $x : \Con{O}\Rightarrow\Con{O}$ represents a variable, then
$x\,\Con{yes}$ should be represent normal form; but under the instantiation of
$x$ by $\prn{\lambda z.z}$, the resulting expression $\prn{\lambda
    z.z}\,\Con{yes}$ should \emph{not} represent a normal form --- the normal form
representing for this expresion should simply be $\Con{yes}$.
It follows that the only way that normal forms could give rise to a presheaf on
$\CX[\IMod]$ is if \emph{variables} gave rise to a presheaf on $\CX[\IMod]$,
but we will see that this does not obtain.

\begin{definition}\label{def:variable}
  The \DefEmph{variables} in a natural model $\Mod$ are defined to be the
  smallest class of morphisms into $\EL[\Mod]$ such that for any
  $\Mor[A]{\Yo{\Gamma}}{\TP[\Mod]}$, the morphism
  $\Mor[\Con{x}_A]{\Yo\prn{\Gamma.A}}{\EL[\Mod]}$ is a variable, and if
  $\Mor[y]{\Yo{\Gamma}}{\EL[\Mod]}$ is a variable than so is the composite
  $\Mor[y\circ \Yo{p_A}]{\Yo\prn{\Gamma.A}}{\EL[\Mod]}$:
  \[
    \begin{tikzpicture}[diagram]
      \SpliceDiagramSquare<sq/>{
        nw/style = pullback,
        nw = \Yo\prn{\Gamma.A},
        sw = \Yo{\Gamma},
        ne = \EL[\Mod],
        se = \TP[\Mod],
        east = \Proj[\Mod],
        south = A,
        north = \Con{x}_A,
        west = \Yo{p_A},
        west/node/style = right,
      }
      \node[left = of sq/sw] (sw) {$\EL[\Mod]$};
      \draw[->,color=RegalBlue] (sq/sw) to node[below] {$y$} (sw);
      \draw[->,magenta] (sq/nw) to node[sloped,above] {$y\circ \Yo{p_A}$} (sw);
    \end{tikzpicture}
  \]
\end{definition}

\begin{problem}[{Variables do not form a presheaf on $\CX[\IMod]$}]
If the collection of variables in the sense of \cref{def:variable} formed a
presheaf on $\CX[\IMod]$, then we could extend it to inductively define a
presheaf normal forms satisfying our desired laws. We might try to define
$\Con{Var}\Sub{\IMod} \in \JDG[\IMod]$ to assign to each context
$\Gamma\in\CX[\IMod]$ the subset of
$\Hom{\JDG[\IMod]}{\Yo{\Gamma}}{\EL[\IMod]}$ spanned by variables in the
sense of \cref{def:variable}. This definition, however, is evidently not
functorial in $\Gamma$: variables are closed under precomposition with
projections $\Mor{\Gamma.A}{\Gamma}$ and certain other maps derived from
these, whereas functoriality in $\CX[\IMod]$ requires closure under
precomposition with arbitrary maps.
\end{problem}

\NewDocumentCommand\RCat{}{\mathscr{R}}

\subsection{Models of variables and the method of computability}\label{sec:variables-and-computability}

Although we have seen that the variables of a natural model $\Mod$ will not
generally arrange themselves into a presheaf on $\CX[\Mod]$. Nonetheless, it is
possible to imagine them forming a presheaf on a \emph{different} category ---
perhaps, intuitively, a wide subcategory of $\CX[\Mod]$ that has fewer
morphisms in it and thus induces a weaker functoriality condition.

If $\Mor[\rho]{\RCat}{\CX[\Mod]}$ represents the inclusion functor of such a
wide subcategory on which the collection of variables forms a presheaf, then
there is some hope for way to state define the collection of normal forms ---
not as a presheaf on $\CX[\Mod]$ but as a presheaf on $\RCat$. Of course, we
must be able to link normal forms of types to the actual types they represent,
so the collection of types $\TP[\Mod]$ must be imported into $\Psh{\RCat}$.
This is easily done, however, by considering its restriction $\rho^*\TP[\Mod]\in\Psh{\RCat}$.

\subsubsection{Models of variables over a natural model}

The situation that we have intuitively described can be made more precise with
the following more general notion of \DefEmph{model of variables}.

\begin{definition}[\citet{bocquet-kaposi-sattler:2021,uemura:2022:coh}]\label{def:model-of-variables}
  A \DefEmph{model of variables} over a natural model $\Mod$ is defined to be a
  natural model $\ModR$ equipped with a homomorphism of natural models
  $\Mor[\rho]{\ModR}{\Mod}$ such that the induced map
  $\Mor{\TP[\ModR]}{\rho^*\TP[\Mod]}$ is an isomorphism.
\end{definition}

In the situation of \cref{def:model-of-variables}, then we may define $\RCat$
to be the underlying category $\CX[\ModR]$. Models of variables over $\Mod$
themselves arrange into a compactly generated (2,1)-category, and so we may consider
the \DefEmph{bi-initial model of variables} over $\IMod$. In this case,
$\EL[\ModR]\in\JDG[\ModR]$ plays the role of the desired presheaf of variables;
indeed, the bi-initiality property here corresponds to the \emph{inductive}
definition of variables (\cref{def:variable}).

\begin{exegesis}
  The purpose of requiring $\Mor{\TP[\ModR]}{\rho^*\TP[\Mod]}$ to be an
  isomorphism is to ensure that variables are classified by the same sorts of
  types as terms, and that the underlying functor
  $\Mor[\rho]{\CX[\ModR]}{\CX[\Mod]}$ is essentially surjective on objects.
  Note that even if $\Mod$ is structured (\eg with function spaces, \etc),
  \cref{def:model-of-variables} refers only to the bare structure of the
  natural model.
\end{exegesis}

\subsubsection{Why is it hard to build a model based on normal forms?}\label{sec:why-hard}

Recalling our construction of a normal forms presentation for free monoids in
\cref{sec:free-monoid-nf}, we should be aiming to construct a natural model
$\Modd$ containing normal forms equipped with a (pseudo-)retraction
$\Mor[P]{\Modd}{\IMod}$ of the induced universal map $\Mor[S]{\IMod}{\Modd}$.
Because we have a suitable notion of variable in $\JDG[\ModR]$ it is tempting
to attempt to define $\CX[\Modd] = \CX[\ModR]$ and then define $\TP[\Modd]$ to
be a presheaf of normal forms of types and $\EL[\Modd]$ to be the presheaf of
normal forms of elements. This proposal will fail almost immediately, however.

\begin{problem}\label{problem:tait:1}
When $\ModR$ is the bi-initial model of variables over $\IMod$, there are
simply not enough morphisms in $\CX[\ModR]$ to build a model $\Modd$ of the
full type theory (\eg with function spaces) over it. For instance, the
function space in $\Modd$ between two global types would necessarily induce an
exponential between their context comprehensions in $\CX[\ModR]$, but this
structure is not present in the bi-initial model of variables.
\end{problem}

A more promising idea to avoid \cref{problem:tait:1} is to let $\CX[\Modd]$ be
a suitable full subcategory of $\JDG[\ModR]$ closed under not only context
comprehension from $\ModR$ but also the image of
$\Mor[\rho^*]{\JDG[\IMod]}{\JDG[\ModR]}$. This doesn't work either, however.

\begin{problem}\label{problem:tait:2}
If we take $\CX[\Modd]$ to be a suitable full subcategory of $\JDG[\ModR]$,
then the resulting model cannot retain enough information about $\IMod$ to induce
a pseudo-retraction $\Mor|->>|[P]{\Modd}{\IMod}$ of the universal map
$\Mor[S]{\IMod}{\Modd}$.\footnote{In fact, a normalization function can be
defined in such a model~\citep{fiore:2022:nbe}, but its correctness cannot be
established without the pseudo-retraction $\Mor|->>|[P]{\Modd}{\IMod}$.}
\end{problem}

There is another problem besides the above with the idea of modeling types and
terms by their normal forms, no matter what ambient category we may choose.
\cref{problem:tait:3} below demonstrates that the problem of normalization for
type theory with function spaces is vastly more difficult than that of (\eg)
the theory of monoids.

\begin{problem}\label{problem:tait:3}
In the presence of function spaces, the collections of normal forms cannot be
used \emph{directly} as a model. Roughly, the problem is that we would need
to define (\eg) an application function that takes a normal form of $u:A\to
  B$ and a normal form of $v:A$ to a normal form of $uv:B$, but this is exactly
the problem we have been trying to solve in the first place --- so we cannot
define this function until our proof is complete.
\end{problem}

\subsubsection{Tait's method of computability}\label{sec:tait}

\cref{problem:tait:1,problem:tait:2,problem:tait:3} were first solved by Bill
Tait, simultaneously, when he introduced the eponymous \DefEmph{method of
  computability}~\citep{tait:1967}, also variously known as \emph{logical
  relations}, \emph{logical predicates}, or the \emph{reducibility
  method}.\footnote{In addition to Tait's original contribution, several other
  authors contributed greatly to the early development (and naming) of this
  concept, including for example
  \citet{girard:1971,martin-lof:1975,martin-lof:itt:1975,martin-lof:1971,plotkin:1973,plotkin:1980,prawitz:1971,statman:1985}.}
Tait's brilliant solution to \cref{problem:tait:1,problem:tait:2}, phrased in
non-categorical language, was to devise a model in which a context is modeled
as a predicate of some kind on a syntactic context; and a substitution is
modeled by a syntactic substitution that preserves the corresponding
predicates. Because every construct in the model is tracked by something
syntactic, there is enough data to define a pseudo-retraction from the model
onto the syntax. By imposing a further condition that the interpretation of
every type be equipped with a \emph{projection} onto normal forms, Tait solves
\cref{problem:tait:3}.

\subsubsection{Freyd's categorical reconstruction of Tait computability}

In \citeyear{freyd:1978}, Peter Freyd rephrased Tait's method into categorical
language as an instance of \DefEmph{Artin gluing} or \DefEmph{recollement},
when he used it to give the first conceptual proof of the existence and
disjunction properties in the free elmentary topos~\citep{freyd:1978}. Of
course, Artin gluing was first introduced in SGA~4 as a way to reconstruct a
topos from complementary open and closed subtopoi. Freyd considered only
gluings along the global sections functor, whereas Tait's original situation
(and ours) requires a more subtle gluing involving the functor that arises from
a model of variables $\Mor[\rho]{\ModR}{\IMod}$. Scenarios of this kind were
first investigated categorically by
\citet{jung-tiuryn:1993,altenkirch-hofmann-streicher:1995,streicher:1998,fiore-simpson:1999,fiore:2002}.
Our own ``synthetic'' approach to Tait's method, to be detailed in
\cref{sec:stc}, is obtained by combining the observations of the cited authors
with the viewpoint of the type theoretic internal language
(\cref{sec:relative-pov}) of glued topoi.

\section{Normalization by gluing for free natural models}

\subsection{Synthetic Tait computability for models of type theory}\label{sec:stc}

We have seen in \cref{sec:variables-and-computability} that the normalization
problem for type theory hinges on the concept of a \emph{variable}, and
introduced a technical notion of ``model of variables'' on a natural model
(\cref{def:model-of-variables}) that can serve as a matrix in which to define
the notion of normal forms. As we pointed out in \cref{sec:why-hard}, this is not enough to
prove that normal forms adequately represent the constructs of the bi-initial
natural model $\IMod$ of type theory.

The fundamental issue, exposed in \cref{problem:tait:1,problem:tait:2}, is that
any normalization model $\Modd$ must be structured with a homomorphism
$\Mor|->>|{\Modd}{\IMod}$, which shall be seen to be a pseudo-retraction of the
induced universal map $\Mor|>->|{\IMod}{\Modd}$ by an application of the
latter's universal property; concretely, this means that both contexts and
substitutions of the normalization model $\Modd$ must not forget the contexts
and substitutions from the bi-initial model to which they pertain.

The reason the pseudo-retraction $\Mor|->>|{\Modd}{\IMod}$ is needed is the
same as in our simpler example for free monoids
(\cref{sec:monoid-normalization}, \cref{cor:scalar-multiplication-injective}):
the universal map $\Mor|>->|{\IMod}{\Modd}$ sends each piece of term $X$ to a
construct of the normalization model from which we might expect to extract a
normal form, and the purpose of the pseudo-retraction is to ensure that the
resulting normal form is a normal form for $X$, rather than a normal form for
some other term.

Bill Tait's solution to these problems was to define models that
compositionally instrument syntactic constructs with additional data, namely
the data of ``normalizability'' or ``computability''. In this section, we will
see how the categorical reconstruction of Tait's method of computability arises
naturally from the idea of formally gluing the constructs of bi-initial model
$\IMod$ along the restriction functor $\Mor[\rho^*]{\JDG[\IMod]}{\JDG[\ModR]}$
onto data valued in the model of variables $\ModR$, leading to a notion of
\DefEmph{computability space} from which a normalization model can be extracted
by means of a certain \DefEmph{functor of points}.

\subsubsection{The topos of computability spaces over a model of variables}

Let $\Mor[\rho]{\ModR}{\Mod}$ be a model of variables over a natural model
$\Mod$. We will first show each aspect of the model of variables translates
into the geometric language of topoi. In particular, the two categories of
presheaves $\JDG[\Mod]$ and $\JDG[\ModR]$ are topoi; in this paper, we are
careful to distinguish the geometrical and algebraic aspects of
topoi~\citep{anel-joyal:2021,vickers:2007,bunge-funk:2006}, so we shall write
$\ClTop{\Mod}$ and $\ClTop{\ModR}$ for the topoi corresponding the two
categories of presheaves respectively.

\begin{notation}
  Given a topos $\XTop$, we will write $\Sh{\XTop}$ for the corresponding
  category; for example, we have $\Sh{\ClTop{\Mod}} =
    \JDG[\Mod]$. We refer to an object of $\Sh{\XTop}$ as a \emph{sheaf on
    $\XTop$}.
\end{notation}

\begin{observation}
  The underlying functor $\Mor[\rho]{\CX[\ModR]}{\CX[\Mod]}$ of our model of
  variables corresponds to an \DefEmph{essential morphism} of topoi
  $\Mor[\brho]{\ClTop{\ModR}}{\ClTop{\Mod}}$ given by the adjoint triple
  induced by base change of presheaves:
  \[
    \begin{tikzpicture}[diagram]
      \node (w) {$\JDG[\Mod]$};
      \node[right = 3cm of w] (e) {$\JDG[\ModR]$};
      \draw[->] (w) to node[upright desc] (c) {$\rho^*$} (e);
      \draw[->,bend right=50] (e) to node[above] (n) {$\rho_!$} (w);
      \draw[->,bend left=50] (e) to node[below] (s) {$\rho_*$} (w);
      \node[between = c and n,yshift = -.15cm] {$\bot$};
      \node[between = c and s,yshift = .15cm] {$\bot$};
    \end{tikzpicture}
  \]
\end{observation}

\begin{lemma}\label{lem:geometric-surjection}
  When $\Mod$ is \emph{democratic} in the sense of \cref{def:democratic}, the
  essential morphism $\Mor[\brho]{\ClTop{\ModR}}{\ClTop{\Mod}}$ induced by the
  model of variables is a \DefEmph{surjection} of topoi.
\end{lemma}

\begin{proof}
  It can be shown that the underlying functor
  $\Mor[\rho]{\CX[\ModR]}{\CX[\Mod]}$ of a model of variables is essentially
  surjective on objects when $\CX[\Mod]$ is democratic. This is enough to see
  that the precomposition functor $\rho^*$ is faithful, and so $\brho =
    \prn{\rho^*\dashv\rho_*}$ is surjective.
\end{proof}

Our goal is to classify a notion of \emph{computability space} that instruments
the constructs of $\Mod$ with data from $\ModR$; the fundamental example of a
computability space would then be the space of normal forms of types: in this
example, one instruments types that live in $\Mod$ with normal forms that live
in $\ModR$.
First we will define precisely what a computability space is, and then we will
observe that that we may construct a topos $\GlTop$ by Artin gluing whose
sheaves are exactly the computability spaces.

\begin{definition}\label{def:computability-space}
  A \DefEmph{computability space} is given by a presheaf $E\in \JDG[\Mod]$
  together with a family of presheaves $\Mor[\pi_E]{\tilde{E}}{\brho^*E}\in
    \JDG[\ModR]$. A morphism from a computability space $\prn{E,\pi_E}$ to a
  computability space $\prn{F,\pi_F}$ is given by a morphism
  $\Mor[f]{E}{F}\in\JDG[\Mod]$ together with a morphism
  $\Mor[\tilde{f}]{\tilde{E}}{\tilde{F}}\in\JDG[\ModR]$ such that the following
  square commutes:
  \[
    \DiagramSquare{
      nw = \tilde{E},
      ne = \tilde{F},
      sw = \brho^*E,
      se = \brho^*F,
      south = \brho^*f,
      north = \tilde{f},
      west = \pi_E,
      east = \pi_F,
    }
  \]
\end{definition}

\begin{construction}
  We define the \DefEmph{topos of computability spaces} to be the following
  co-comma object in the bicategory of Grothendieck topoi below:
  \begin{equation}\label[diagram]{diag:cocomma}
    \begin{tikzpicture}[diagram,baseline=(sw.base)]
      \SpliceDiagramSquare{
        nw = \ClTop{\ModR},
        sw = \ClTop{\ModR},
        ne = \ClTop{\Mod},
        se = \GlTop,
        north = \brho,
        east = \bj,
        south = \bi,
        west/style = {double},
        east/style = {exists,open immersion},
        south/style = {exists,closed immersion},
      }
      \node[between = sw and ne] {$\xRightarrow{\pi}$};
    \end{tikzpicture}
  \end{equation}
\end{construction}

In \cref{diag:cocomma}, the morphism
$\Mor|open immersion|[\bj]{\ClTop{\Mod}}{\GlTop}$ is a \DefEmph{open immersion}
of topoi, and the $\Mor|closed immersion|[\bi]{\ClTop{\ModR}}{\GlTop}$ is its
complementary \DefEmph{closed immersion}. This gluing can be computed more
explicitly as a \DefEmph{closed mapping cylinder} using the Sierpi\'nski interval
$\mathbb{S} = \brc{\Mor{\bullet}{\circ}}$, as in \citet{johnstone:topos:1977}:
\begin{equation}\label[diagram]{diag:closed-mapping-cylinder}
  \begin{tikzpicture}[diagram,baseline=(sq/sw.base)]
    \SpliceDiagramSquare<sq/>{
      nw = \ClTop{\ModR},
      sw = \ClTop{\ModR}\times\mathbb{S},
      ne = \ClTop{\Mod},
      se = \GlTop,
      north = \brho,
      east = \bj,
      south = \bk,
      west/node/style = right,
      west = \prn{\ClTop{\ModR},\circ},
      west/style = {open immersion, bend left=30},
      east/style = open immersion,
      se/style = pushout,
      width = 3cm,
    }
    \draw[closed immersion, bend right=30] (sq/nw) to node[left] (l) {$\prn{\ClTop{\ModR},\bullet}$} (sw);
    \node[right = 1.08cm of l] {$\Rightarrow$};
    \node(G) [left = 3cm of sq/sw] {$\GlTop$};
    \draw[->] (sq/sw) to node[below] {$\bk$} (G);
    \draw[closed immersion,bend right] (nw) to node[sloped,above] {$\bi$} (G);
  \end{tikzpicture}
\end{equation}

\begin{observation}\label{obs:comma-computation}
  Under the geometry--algebra duality, the co-comma topos $\GlTop$ corresponds
  to the comma category $\Sh{\GlTop}\simeq \prn{\JDG[\ModR]\downarrow
      \brho^*}$, \ie the \DefEmph{Artin gluing} of $\brho^*$. Therefore, sheaves on
  $\GlTop$ are the same thing as computability spaces \emph{qua}
  \cref{def:computability-space}, and morphisms between sheaves are exactly
  morphisms of computability spaces.
\end{observation}

\begin{lemma}
  Both the open and closed immersions are essential, \ie we have additional
  (necessarily fully faithful) left adjoints $\bj_!\dashv \bj^*$ and
  $\bi_!\dashv \bi^*$.
\end{lemma}

\begin{proof}
  That the open immersion is essential follows from the fact that $\JDG[\ModR]$
  has an initial object; that the closed immersion is essential follows from
  the fact that $\Mor[\brho]{\ClTop{\ModR}}{\ClTop{\Mod}}$ is essential.
  Finally, in an adjoint triple $F\dashv G\dashv H$, the leftmost adjoint is
  fully faithful if and only if the rightmost one is.
\end{proof}

\begin{lemma}
  In fact, we have a further right adjoint $\bj_*\dashv \bj^!$.
\end{lemma}

\begin{proof}
  Because $\brho^*$ has a right adjoint $\brho_*$.
\end{proof}

\begin{exegesis}
  Under the identification of sheaves on $\GlTop$ with computability spaces, we
  may examine the behavior of all the adjoints $\bj_!\dashv\bj^*\dashv\bj_*\dashv\bj^!$ and $\bi_!\dashv\bi^*\dashv\bi_*$. We first describe the inverse image functors:
  \begin{align*}
    \bj^*\prn{E,\pi_E} &= E\\ 
    \bi^*\prn{E,\pi_E} &= \operatorname{dom}\pi_E
  \end{align*}
  
  Next we compute all the other adjoints.
  \begin{align*}
    \color{gray}\bj^*\dashv\normalcolor \bj_*E 
    &= \prn{E, \Mor[\Idn{\brho^*E}]{\brho^*E}{\brho^*E}}
    \\ 
    \color{gray}\bj^*\vdash\normalcolor \bj_!E 
    &= \prn{E,\Mor[!\Sub{\brho^*E}]{\ObjInit{\JDG[\ModR]}}{\brho^*E}}
    \\
    \color{gray}\bj_*\dashv\normalcolor \bj^!G 
    &= \brho_*\bi^*G\times\Sub{\brho_*\brho^*\bj^*G} \bj^*G
    \\\\ 
    \color{gray}\bi^*\dashv\normalcolor\bi_*R 
    &= \prn{\ObjTerm{\JDG[\Mod]},\Mor[!\Sub{R}]{R}{\brho^*\ObjTerm{\JDG[\Mod]}}}
    \\ 
    \color{gray}\bi^*\vdash\normalcolor\bi_!R
    &= \prn{\brho_!R,\Mor[\eta\Sub{R}]{R}{\brho^*\brho_!R}}
  \end{align*}
\end{exegesis}

\begin{exegesis}
  The additional left adjoint $\bi_!\dashv\bi^*$ will play an important role; it
  is uniquely determined by the property of sending representables $\Gamma$ to a
  space of ``variable renamings'' for $\Gamma$ in $\Sh{\GlTop}$. If we think of
  $\Gamma$ as a context, then an element of $\bi_!\Yo{\Gamma}$ can be thought of
  as representing a sequence of variables that can be substituted for those
  classified by $\Gamma$.
\end{exegesis}

\begin{lemma}\label{lem:display-of-variables-square}
  The following square commutes up to isomorphism:
  \[
    \begin{tikzpicture}[diagram]
      \node (nw) {$\CX[\ModR]$};
      \node[below = of nw] (sw) {$\CX[\Mod]$};
      \node[right = of nw] (nc) {$\JDG[\ModR]$};
      \node[right = of nc] (ne) {$\Sh{\GlTop}$};
      \node[below = of ne] (se) {$\JDG[\Mod]$};
      \draw[->] (nw) to node[above] {$\Yo[\CX[\ModR]]$} (nc);
      \draw[->] (nc) to node[above] {$\bi_!$} (ne);
      \draw[->] (nw) to node[left] {$\rho$} (sw);
      \draw[->] (sw) to node[below] {$\Yo[\CX[\Mod]]$} (se);
      \draw[->] (ne) to node[right] {$\bj^*$} (se);
    \end{tikzpicture}
  \]
\end{lemma}

\subsubsection{Recollement of computability spaces}\label{sec:recollement}

SGA~4 explains how the construction of $\GlTop$ by gluing along
$\Mor[\brho]{\ClTop{\ModR}}{\ClTop{\Mod}}$ corresponds, in reverse, to the
\emph{partitioning} of the topos $\GlTop$ into complementary open and closed
subtopos~\citep{sga:4}. Under the latter viewpoint, the open and closed
immersions become identified with the \emph{inclusion} of the corresponding
open and closed subtopoi. We will use this perspective to develop a more
convenient language for constructing computability spaces intrinsically in the
language of $\Sh{\GlTop}$ without bothering with the complex families of
presheaves by which we originally defined computability spaces
(\cref{def:computability-space}).

\begin{definition}[Opens of a topos]
  An \emph{open} of a topos $\XTop$ is defined to be a subterminal sheaf on
  that topos, \ie a subobject of $\ObjTerm{\Sh{\XTop}}$. We will write
  $\Opns{\XTop}\subseteq\Sh{\XTop}$ for the poset (frame, in fact) of opens of
  $\XTop$.
\end{definition}

\begin{definition}
  Let $U\in\Opns{\XTop}$ be an open of a topos $\XTop$; then a sheaf $E$ is
  called \DefEmph{$U$-modal} when the canonical map $\Mor{E}{E^U}$ is an
  isomorphism.  Conversely, a sheaf $E$ is called \DefEmph{$U$-connected} when
  the projection map $\Mor{E\times U}{U}$ is an isomorphism.
\end{definition}

\begin{fact}[$U$-modal and $U$-connected reflection]
  For an open $U\in\Opns{\XTop}$, the full subcategories of $\Sh{\XTop}$
  spanned by $U$-modal and $U$-connected sheaves are reflective.
  \begin{enumerate}
    \item The $U$-modal reflection of a sheaf $X\in\Sh{\XTop}$ is given by the exponential $X^U$.
    \item The $U$-connected reflection of a sheaf $X\in\Sh{\XTop}$ is given by the pushout of the product span $X\xleftarrow{\pi_1}X\times U\xrightarrow{\pi_2}U$.
  \end{enumerate}
\end{fact}

The $U$-modal and $U$-connected reflections preserve finite limits.
Therefore, the full subcategories of $U$-modal and $U$-connected sheaves
present \emph{subtopoi}; the subtopos of $U$-modal sheaves is referred to as
the \DefEmph{open subtopos} determined by $U$ and the subtopos of
$U$-connected sheaves is referred to as the \DefEmph{closed subtopos}
determined by $U$.

\begin{fact}[Open subtopos as slice]
  Based on the definition of the $U$-modal reflector, it is not difficult to
  see that the slice category $\Sl*{\Sh{\XTop}}{U}$ may be canonically
  identified with the category of sheaves on the open subtopos of $\XTop$
  determined by $U$.
\end{fact}

\begin{observation}[Lawvere--Tierney topologies]
  The open and closed subtopi can equivalently be described by Lawvere--Tierney
  topologies on $\XTop$, which simply internalize the reflectors as endomaps of
  the subobject classifier.
  \begin{enumerate}

    \item The topology of the open subtopos is given by the map $\Sl{j}{U} \phi
            = U \Rightarrow \phi$.

    \item The topology of the closed subtopos is given by $j\Sub{\setminus
              {U}}\phi = U\lor \phi$

  \end{enumerate}
\end{observation}

\begin{construction}[Recollement of the topos of computability spaces]\label{con:topos-recollement}
  What we take from SGA~4~\citep{sga:4} is that up to categorical equivalence,
  we may reconstruct the gluing data for our own topos $\GlTop$ from a certain
  open $\P\in\Opns{\GlTop}$, which can be equivalently described by either $\P
    = \bi_*\bot$ or $\P=\bj_!\top$. As a subterminal computability space, $\P$ is
  the family $\prn{\ObjTerm{\JDG[\Mod]},
      \Mor{\ObjInit{\JDG[\ModR]}}{\brho^*\ObjTerm{\JDG[\Mod]}}}$.
  It is then not difficult to see the following:
  \begin{enumerate}

    \item The essential image of $\EmbMor[\bj_*]{\JDG[\Mod]}{\Sh{\GlTop}}$ is
          exactly the full subcategory spanned by $\P$-modal computability spaces.
          Under this identification, the $\P$-modal reflection takes a
          computability space $X$ to $\bj^*X$.

    \item The essential image of $\EmbMor[\bi_*]{\JDG[\ModR]}{\Sh{\GlTop}}$ is
          exactly the full subcategory spanned by $\P$-connected computability
          spaces. Under this identification, the $\P$-connected reflection takes a
          computability space $X$ to $\bi^*X$.

  \end{enumerate}

  Finally, we have a functor from $\P$-modal sheaves to $\P$-connected sheaves
  taking $\P$-modal $E$ to the $\P$-connected reflection of $E$. Under the
  identifications above, this functor is isomorphic to
  $\Mor[\brho^*]{\JDG[\Mod]}{\JDG[\ModR]}$. Thus the open $\P\in\Opns{\GlTop}$
  controls \emph{all} the gluing data of $\GlTop$ except for the additional
  fact that $\brho^*$ happens to be the inverse image component of an essential
  morphism of topoi.
\end{construction}

\begin{fact}[Recollement of computability spaces]\label{fact:mini-recollement}
  Just as \cref{con:topos-recollement} shows that the topos of computability
  spaces can be reconstructed from the induced open and closed subtopoi,
  something similar can be said of each individual computability space. In
  particular, for any $X\in\Sh{\GlTop}$ the following square is always cartesian:
  \[
    \DiagramSquare{
      nw/style = pullback,
      width = 2.5cm,
      nw = X,
      sw = \bj_*\bj^*X,
      ne = \bi_*\bi^*X,
      se = \bi_*\bi^*\bj_*\bj^*X,
      west = \eta\Sub{X}\Sup{\bj_*\bj^*},
      south = \eta\Sub{bj_*\bj^*X}\Sup{\bi_*\bi^*},
      east = \bi_*\bi^*\eta\Sub{X}\Sup{\bj_*\bj^*},
      north = \eta\Sub{X}\Sup{\bi_*\bi^*},
    }
  \]
\end{fact}

The import of \cref{fact:mini-recollement} is that it shows that any sheaf on
$\GlTop$ can be constructed entirely in terms of (left exact, idempotent)
monads on $\Sh{\GlTop}$ without bringing either $\JDG[\Mod]$ nor $\JDG[\ModR]$
into the picture.

\subsubsection{The internal language of computibility spaces}

Although we will not expose them all in this paper, there are a number of
somewhat technical constructions of computability spaces that must ultimately
be carried out. As these constructions are \emph{relative} in nature and must
constantly move between slices of $\Sh{\GlTop}$, we may simplify things
considerably by recalling from \cref{sec:relative-pov} that \DefEmph{type
  theoretic internal languages} are the appropriate linguistic foundation for the
relative point of view.

It happens that \emph{all} the constructions of \cref{sec:recollement} are
stable under slicing, and can therefore be incorporated in a type theoretic
internal language. As a result, we may rephrase the results of
\cref{sec:recollement} as statements in the internal language of $\Sh{\GlTop}$
by adopting the following single postulate:

\begin{postulate}\label{postulate:open}
  There exists a proposition $\P$, \ie a type $\P$ satisfying the condition
  that every two of its elements are equal. We will write $\OpMod$ for the
  reflection of $\P$-modal types; we will write $\ClMod$ for the reflection of
  $\P$-connected types. We additionally assume that $\OpMod{A} = A$ strictly
  when $\P=\top$.\footnote{This final assumption can be removed, but it is
    convenient for our presentation.}
\end{postulate}

\begin{notation}[Subuniverses of modal types]
  Given a universe $\UU$, we will write $\UU\Sub{\OpMod}$ and $\UU\Sub{\ClMod}$
  for the subuniverses spanned by $\P$-modal and $\P$-connected types
  respectively. Note that unlike in univalent
  foundations~\citep{rijke-shulman-spitters:2020}, it is not the case that
  $\UU\Sub{\ClMod}$ is itself $\P$-connected nor that $\UU\Sub{\OpMod}$ is
  $\P$-modal.
\end{notation}

The language of type theory extended by \cref{postulate:open} is referred to by
\citet{sterling:2021:thesis} as \DefEmph{synthetic Tait computability}, because
it generates as if from the void an abstract form of Tait's computability out
of the dynamics of $\P$-modal and $\P$-connected types in the internal
language, as these correspond under the computability spaces interpretation to
the syntactic components and their semantic instrumentations respectively.
Indeed, the internal / type theoretic version of the recollement of
computability spaces (\cref{fact:mini-recollement}) is the following
\cref{obs:synth-recollement}, formally deducible in synthetic Tait
computability.

\begin{observation}[Recollement of computability spaces, synthetically]\label{obs:synth-recollement}
  For any type $X$, the canonical ``fracture function'' defined below is an
  isomorphism:
  \begin{align*}
    X & \to \Sum{x_0:\OpMod{X}}\Compr{x_1 : \ClMod{X}}{
      \eta\Sup{\ClMod}\Sub{\OpMod{X}}x_0 =\Sub{\ClMod\OpMod{X}} \ClMod\eta\Sup{\OpMod}_X x_1
    }
    \\
    x & \mapsto \prn{
      \eta\Sup{\OpMod}_Xx,
      \eta\Sup{\ClMod}_Xx
    }
  \end{align*}
\end{observation}

\begin{notation}\label{notation:bag}
  For any $\P$-modal type $A$, the map $k_A : A\to\prn{\P\to A}$ is invertible
  by definition. We will permit the following abuse of notation: when
  constructing an element of a $\P$-modal type $A$, we will write $\bags{a}$ to
  mean $k_A\Sup{-1}\prn{\lambda{\_}. a}$. Thus inside the delimiter, we
  implicity bind a variable $\_:\P$.
\end{notation}

\begin{notation}[Extension types]
  Let $A:\UU$ be a type and let $\_:\P\vdash a:A$ a be partial element of $A$.
  Then we shall write $\Ext{A}{a}$ for the subtype $\Compr{x:A}{\forall \_:\P.
      x = a}$, called the \DefEmph{extension type} after
  \citet{riehl-shulman:2017}.
\end{notation}

\begin{definition}[Vertical maps]
  If $A$ and $B$ are types such that $\OpMod\prn{A=B}$ holds, then we define a
  \DefEmph{vertical map} from $A$ to $B$ to be a function of the form
  $f:\Ext{\prn{A\to B}}{\lambda x.x}$.
\end{definition}

We refine the synthetic recollement of computability spaces
(\cref{obs:synth-recollement}) with a special type connective to build
computability spaces from their $\P$-modal and $\P$-connected components.

\begin{postulate}[Strict gluing~\citep{gratzer-shulman-sterling:2022:universes,sterling-harper:2022}]\label{postulate:strict-gluing}
  On any of the ambient universes $\UU$, we have a \DefEmph{strict gluing
    operation} that takes a $\P$-modal type $A:\UU\Sub{\OpMod}$ and a family of
  $\P$-connected types $B:A\to \UU\Sub{\ClMod}$ to a type
  $\GlTp{x:A}{Bx}:\Ext{\UU}{A}$ and an isomorphism $\Con{glue}\Sub{A,B} :
    \Ext{\Sum{x:A}Bx\cong \GlTp{x:A}{Bx}}{\pi_1}$.
\end{postulate}

\begin{notation}[Gluing projections and constructor]
  For $g:\GlTp{x:A}{Bx}$, the first projection
  $\pi_1\Con{glue}\Sub{A,B}\Sup{-1}g$ of $g:\GlTp{x:A}{Bx}$ can already be
  written $\bags{g}:A$. We shall write $\underline{g}:B\,\bags{g}$ for the
  second projection $\pi_2\Con{glue}\Sub{A,B}\Sup{-1}g$.  Given $a:A$ and $b:Ba$
  we shall write $\Glue{a}{b}$ for the element $\Con{glue}\Sub{A,B}\,\prn{a,b}$.
\end{notation}

\subsubsection{Internalizing the model of variables}

The model of variables $\Mor[\rho]{\ModR}{\Mod}$ can be internalized into the
synthetic Tait computability of $\Sh{\GlTop}$ by additional postulates.

\begin{postulate}[The base model]\label{postulate:base-model}
  There is a $\P$-modal universe $\prn{\TP,\EL}$ such that for each code
  $A:\TP$ the type $\EL\,{A}$ is $\P$-modal. Moreover, $\TP$ is closed under
  function spaces as well as a base type $\Con{O}:\TP$ and two constants
  $\Con{yes},\Con{no}:\EL\,\Con{O}$.
\end{postulate}

\begin{postulate}[The model of variables]\label{postulate:var} There is an
  additional decoding family $\Var$ on $\TP$ such that for each $A:\TP$, we
  have $\OpMod\prn{\Var\,{A} = \EL\,A}$ or equivalently $\P\Rightarrow
    \Var\,{A}=\EL\,A$.
\end{postulate}

\subsubsection{The computability space of normal forms}

With \cref{postulate:base-model,postulate:var} in hand, it becomes possible to
define a space of normal forms for types by means of an indexed inductive
definition --- or, for the more categorically inclined, as the initial algebra
for a certain polynomial endofunctor on a slice of the ambient universe \`a la
Fiore~\citep{fiore:2002}.  In what follows, we will let $\UU$ be a sufficiently
large universe in the ambient type theory so as to classify each $\EL\,A$ and
$\Var{A}$.

\begin{definition}\label{def:normal-form-algebra}
  A $\UU$-small \emph{normal form algebra} is defined to be a series
  of constants whose sorts we shall specify forthwith. First, a normal form
  algebra requires a sort $\Con{NfTp}$ of normal forms of types, and for each
  type $A:\TP$ a pair of sorts $\NfTm{A},\NeTm{A}$ classifying
  \DefEmph{normal} and \DefEmph{neutral} forms of elements of $A$.

  \iblock{
    \mrow{\NfTp : \Ext{\UU}{\TP}}
    \mrow{\NfTm : \Prod{A:\TP}\Ext{\UU}{\EL\,A}}
    \mrow{\NeTm : \Prod{A:\TP}\Ext{\UU}{\EL\,A}}
  }

  Next we require constructors for the normal forms of each type:

  \iblock{
    \mrow{
      \Con{nfO}
      :
      \Ext{\NfTp}{\Con{O}}
    }
    \mrow{
      \Con{nfFun}
      :
      \Prod{A,B:\NfTp}
      \Ext{\NfTp}{A\Rightarrow B}
    }
  }

  Finally we require constructors for neutral and normal forms of terms.

  \iblock{
    \mrow{
      \Con{neVar} : \Prod{A:\NfTp} \Prod{x:\Var\bags{A}} \Ext{\NeTm{A}}{x}
    }
    \mrow{
      \Con{neApp} :
      \Prod{A,B:\NfTp}
      \Prod{f:\NeTm{\bags{A\Rightarrow B}}}
      \Prod{x:\NfTm{\bags{A}}}
      \Ext{\NeTm{\bags{B}}}{f x}
    }

    \row

    \mrow{
      \Con{nfNeO} :
      \Prod{x:\NeTm{\Con{O}}}\Ext{\NfTm{\Con{O}}}{x}
    }

    \mrow{
      \Con{nfYes} :
      \Ext{\NfTm{\Con{O}}}{\Con{yes}}
    }

    \mrow{
      \Con{nfNo} :
      \Ext{\NfTm{\Con{O}}}{\Con{no}}
    }

    \mrow{
      \Con{nfLam} :
      \Prod{A,B:\NfTp}
      \Prod{f:\Var\,A\to\NfTm{\bags{B}}}
      \Ext{\NfTm{\bags{A\Rightarrow B}}}{\lambda x. f x}
    }

  }
\end{definition}

\begin{definition}
  Let $\mathfrak{M}$ and $\mathfrak{N}$ be two $\UU$-small normal form
  algebras. A \DefEmph{morphism of normal form algebras} from
  $\Mor[H]{\mathfrak{M}}{\mathfrak{N}}$ is given by functions between
  the three carriers

  \iblock{
    \mrow{
      H\Sub{\NfTp} :
      \Prod{A:\NfTp[\mathfrak{M}]}\Ext{\NfTp[\mathfrak{N}]}{A}
    }
    \mrow{
      H\Sub{\NfTm} :
      \Prod{A:\TP}
      \Prod{x:\NfTm[\mathfrak{M}]{A}} \Ext{\NfTm[\mathfrak{N}]{A}}{x}
    }
    \mrow{
      H\Sub{\NeTm} :
      \Prod{A:\TP}
      \Prod{x:\NeTm[\mathfrak{M}]{A}} \Ext{\NeTm[\mathfrak{N}]{A}}{x}
    }
  }

  \noindent
  that preserve all the operations of the normal form algebra in the sense of
  the following representative equations:

  \iblock{
    \mrow{
      H\Sub{\NfTp}\,\Con{nfO}\Sub{\mathfrak{M}} = \Con{nfO}\Sub{\mathfrak{N}}
    }
    \mrow{
      H\Sub{\NfTp}\,\prn{\Con{nfFun}\Sub{\mathfrak{M}}\,A\,B} =
      \Con{nfFun}\Sub{\mathfrak{N}}\,\prn{
        H\Sub{\NfTp}A,
        H\Sub{\NfTp}B
      }
    }
    \mhang{
      H\Sub{\NeTm}\,\bags{B}\,\prn{
        \Con{neApp}\Sub{\mathfrak{m}}\,
        A\,B\,f\,u
      }
      =
    }{
      \mrow{
        \Con{neApp}\Sub{\mathfrak{N}}\,
        \prn{H\Sub{\NfTp}A}\,
        \prn{H\Sub{\NfTp}B}\,
        \prn{
          H\Sub{\NeTm}\,\bags{A\Rightarrow B}\,f
        }\,
        \prn{
          H\Sub{\NfTm}\,\bags{A}\,u
        }
      }
    }
    \mrow{\ldots}
  }
\end{definition}

\begin{lemma}\label{lem:ext-sl-equiv}
  For $X:\UU\Sub{\OpMod}$, denote by $\Ext{\UU}{X}$ the category whose morphisms
  are given by \emph{vertical} maps. The functor
  $\Mor{\Ext{\UU}{X}}{\Sl*{\UU\Sub{\ClMod}}{X}}$ sending each $A:\Ext{\UU}{X}$
  to the family $\prn{x:X}\mapsto \Ext{A}{x}$ is an equivalence.
\end{lemma}
\begin{proof}
  This follows from \cref{postulate:strict-gluing}.
\end{proof}

\begin{lemma}
  The exists an \DefEmph{initial normal form algebra}.
\end{lemma}

\begin{proof}
  Evidently, the initial normal form algebra would be the initial
  algebra for a certain endofunctor $\mathfrak{F}$ on the product category
  \[
    \Ext{\UU}{\TP}\times \prn{\Prod{A:\TP}\Ext{\UU}{\EL\,{A}}}^2
  \]
  if such an initial algebra exists. By
  \cref{lem:ext-sl-equiv} and the disjointness property of
  sums we may equivalently present the category above as a
  slice of $\UU\Sub{\ClMod}$:
  \begin{align*}
     &
    \Ext{\UU}{\TP[\Mod]}\times
    \prn{\Prod{A:\TP[\Mod]}\Ext{\UU}{\EL\,{A}}}^2
    \\
     & \quad\simeq
    \Sl*{\UU\Sub{\ClMod}}{\TP}
    \times
    \prn{
      \Prod{A:\TP}
      \Sl*{\UU\Sub{\ClMod}}{\EL\,A}
    }^2
    \tag{\cref{lem:ext-sl-equiv}}
    \\
     & \quad\simeq
    \Sl*{\UU\Sub{\ClMod}}{\TP}
    \times
    \prn{\Sl*{\UU\Sub{\ClMod}}{\prn{\Sum{A:\TP}{\EL\,A}}}}^2
    \tag{disjointness}
    \\
     & \quad\simeq
    \Sl*{\UU\Sub{\ClMod}}{
      \prn{
        \TP + 2\times \Sum{A:\TP}{\EL\,A}
      }
    }
    \tag{disjointness}
  \end{align*}

  Under this identification, the endofunctor $\mathfrak{F}$ can be seen to be
  polynomial. Because $\UU\Sub{\ClMod}$ has W-types and equality types, the
  initial algebra exists~\citep{gambino-kock:2013}.
\end{proof}

\subsubsection{Injectivity of normal type constructors}

Let $\mathfrak{M}$ be the initial normal form algebra.

\begin{construction}
  Let $\Con{isFun} : \NfTp[\mathfrak{M}]\to \UU\Sub{\ClMod}$ be the family
  sending each $A$ to
  $\ClMod\Compr{\prn{B,C}:\NfTp[\mathfrak{M}]^2}{A=\Con{nfFun}\Sub{\mathfrak{M}}\,B\,C}$.
  We will define an auxiliary normal form algebra $\mathfrak{P}$ such that
  $\NfTp[\mathfrak{P}]$ associates to each normal type
  $A:\NfTp[\mathfrak{M}]$ a type $A'$ equipped with a map into
  $\Con{isFun}\,A$.  Of course, this description evokes the \emph{Artin gluing}
  $\NfTp[\mathfrak{M}]\downarrow \Con{isFun}$ when we view
  $\NfTp[\mathfrak{M}]$ as a discrete category:
  \[
    \DiagramSquare{
      sw = \NfTp[\mathfrak{M}],
      se = \UU\Sub{\ClMod},
      south = \Con{isFun},
      ne = \UU\Sub{\ClMod}\Sup{\to},
      east = \Con{cod},
      nw = \NfTp[\mathfrak{P}],
      nw/style = pullback,
      west/style = {->,exists},
      north/style = {->,exists},
    }
  \]

  More explicitly, we define $\NfTp[\mathfrak{P}]$ and the rest of the
  algebra as follows:

  \iblock{
    \mrow{
      \NfTp[\mathfrak{P}]
      =
      \Sum{A:\NfTp[\mathfrak{M}]}
      \Sum{A':\UU\Sub{\ClMod}}
      \prn{A'\to \Con{isFun}\,A}
    }

    \mrow{
      \NfTm[\mathfrak{P}]{\prn{A,A',\_}} =
      \NfTm[\mathfrak{M}]{\bags{A}}
    }

    \mrow{
      \Con{nfO}\Sub{\mathfrak{P}}
      =
      \prn{
        \Con{nfO}\Sub{\mathfrak{M}},
        \P,
        \lambda\_.\star
      }
    }

    \mrow{
      \Con{nfFun}\Sub{\mathfrak{P}}\,A\,B =
      \prn{
        \Con{nfFun}\Sub{\mathfrak{M}}\,\prn{\pi_1A}\,\prn{\pi_1B},
        \top,
        \lambda{\_}.
        \eta\Sub{\ClMod}\prn{\pi_1{A},\pi_1{B}}
      }
    }

    \mrow{\ldots}
  }

  We evidently have a homomorphism of algebras
  $\Mor[\pi]{\mathfrak{P}}{\mathfrak{M}}$ forgetting the additional
  information. As $\mathfrak{M}$ is initial, this projection homomorphism in
  fact has a (unique) section $\Mor[I]{\mathfrak{M}}{\mathfrak{P}}$:
  \[
    \begin{tikzpicture}
      \node (M) {$\mathfrak{M}$};
      \node (P) [right = of M] {$\mathfrak{P}$};
      \node (M') [below = of P] {$\mathfrak{M}$};
      \draw[exists,->] (M) to node [above] {$I$} (P);
      \draw[->] (P) to node [right] {$\pi$} (M');
      \draw[double] (M) to (M');
    \end{tikzpicture}
  \]
\end{construction}

\begin{lemma}[Modal injectivity of normal form constructors]\label{lem:modal-injectivity}
  The functorial map
  $\Mor[\ClMod\Con{nfFun}\Sub{\mathfrak{M}}]{
      \ClMod\NfTp[\mathfrak{M}]^2
    }{\ClMod\NfTp[\mathfrak{M}]}$
  is a monomorphism.
\end{lemma}
\begin{proof}
  The claim is equivalent to the following formula:
  \[
    \forall P, Q : \NfTp[\mathfrak{M}]^2.\
    \Con{nfFun}\Sub{\mathfrak{M}}P=\Con{nfFun}\Sub{\mathfrak{M}}Q
    \to
    \eta\Sub{\ClMod}P=\eta\Sub{\ClMod}Q
  \]

  Fix $P$ and $Q$ such that
  $\Con{nfFun}\Sub{\mathfrak{M}}F=\Con{nfFun}\Sub{\mathfrak{M}}Q$. Considering the
  action of the universal map $\Mor[I]{\mathfrak{M}}{\mathfrak{P}}$ on this
  section, we have:
  \[
    \DiagramSquare{
      width = 4cm,
      west/style = double,
      east/style = double,
      north/style = double,
      south/style = double,
      nw = I\Sub{\NfTp}\prn{\Con{nfFun}\Sub{\mathfrak{M}}P},
      sw = I\Sub{\NfTp}\prn{\Con{nfFun}\Sub{\mathfrak{M}}Q},
      ne = {
          \prn{
            \Con{nfFun}\Sub{\mathfrak{M}}P,
            \top,
            \lambda{*}.\eta\Sub{\ClMod}P
          }
        },
      se = {
          \prn{
            \Con{nfFun}\Sub{\mathfrak{M}}Q,
            \top,
            \lambda{*}.\eta\Sub{\ClMod}Q
          }
        },
    }
  \]

  Thus by projection, we have $\eta\Sub{\ClMod}P=\eta\Sub{\ClMod}Q$.
\end{proof}

\subsubsection{The universe of normalization spaces}\label{sec:u-norm-space}

Our goal has been to define a natural model of type theory lying over the
bi-initial model $\IMod$ in which normal forms can be projected from the
interpretations of types, following our discussion of \citet{tait:1967} in
\cref{sec:tait}. Tait's idea, which we will realize in a more technical form
here, is to let the semantic universe of the normalization model assign a
(vertical) projection map from each kind of semantic object into the
corresponding space of normal forms. In order to close such a universe under
function spaces, Tait noticed that it was necessary to have a vertical map
\emph{into} every semantic type from the space of neutral forms of elements of
that type.
In this section, we aim to define a universe of \DefEmph{normalization spaces},
or computability spaces that are equipped with the structure described above.

\begin{definition}
  A \DefEmph{normalization space} $A$ is given by the following data:
  \begin{enumerate}
    \item a normal form $\ReifyTp{A}:\NfTp$;
    \item a type $\EL\NS{A}:\Ext{\UU}{\EL\,\ReifyTp{A}}$;
    \item a ``reflection'' map $\Reflect{A} : \Prod{x:\NeTm\,\bags{\ReifyTp{A}}} \Ext{\EL\NS{A}}{x}$;
    \item a ``reification'' map $\Reify{A} : \Prod{x:\EL\NS{A}}\Ext{\NfTm\,\bags{\ReifyTp{A}}}{x}$.
  \end{enumerate}
  Of course, the reflection and reification maps can be stated as a sequence of
  \emph{vertical} maps $\NeTm\,\bags{\ReifyTp{A}}\to \EL\NS{A}\to
    \NfTm\,\bags{\ReifyTp{A}}$.
\end{definition}

\begin{construction}[The universe of normalization spaces]
  By \cref{postulate:strict-gluing}, we may define a type $\TP\NS$ of
  normalization spaces such that $\OpMod\prn{\TP\NS = \TP}$ strictly. Thus we
  have a universe $\NN = \prn{\TP\NS,\EL\NS}$ that restricts under $\OpMod$ to
  $\prn{\TP,\EL}$.
\end{construction}

\subsubsection{Closure of normalization spaces under connectives}

We can close the universe of normalization spaces (\cref{sec:u-norm-space})
under the connectives that we have postulated on $\TP$ in such a way that they
restrict exactly to these under $\OpMod$.

\begin{construction}[The function space in normalization spaces]
  For function spaces, we must define the following map (as well as
  corresponding maps for $\lambda$-abstraction and application):
  \[
    \prn{\Rightarrow\NS} : \Ext{\prn{\TP\NS\times\TP\NS\to \TP\NS}}{\prn{\Rightarrow}}
  \]

  Given two normalization spaces $A,B:\TP\NS$ we must define a normalization
  space $\prn{A\Rightarrow\NS B}:\Ext{\TP\NS}{A\Rightarrow B}$. Below, we
  describe how to construct this space:

  \begin{enumerate}

    \item To define the normal form $\ReifyTp{\prn{A\Rightarrow\NS B}} :
            \Ext{\NfTp}{A\Rightarrow B}$, we use the normal forms of $A$ and of $B$ to
          construct $\Con{nfFun}\,{\ReifyTp{A}}\,{\ReifyTp{B}}$.

    \item To define the type $\EL\NS\prn{A\Rightarrow\NS B}$ over
          $\EL\,\prn{A\Rightarrow B}$, we will use the function space $\EL\NS{A} \to
            \EL\NS{B}$. Note that this restricts only up to isomorphism to
          $\EL\,\prn{A\Rightarrow B}$, but that this can be corrected using
          \cref{postulate:strict-gluing}. Therefore, we will not belabor the point
          further in our informal explanation.

    \item To define the reflection map $\Reflect{A\Rightarrow\NS B}$, we are
          given a neutral form $f:\NeTm\,\prn{A \Rightarrow B}$ and an element
          $x:\EL\NS A$ and must produce an element $\prn{\Reflect{A\Rightarrow\NS
                B}f}x : \Ext{\EL\NS B}{f x}$. Applying the reflection map for $B$, it
          suffices to give a neutral form in $\Ext{\NeTm\,\bags{B}}{f x}$; applying
          the neutral application constructor $\Con{neApp}$, we need only a normal
          form in $\Ext{\NfTm\,\bags{A}}{x}$, why we obtain by reification at $A$.
          All in all we have:
          \[
            \prn{\Reflect{A\Rightarrow\NS B}f}x =
            \Reflect{B}{\Con{neApp}\,{\ReifyTp{A}}\,{\ReifyTp{B}}\,{f}\,\prn{\Reify{A}x}}
          \]

    \item To define the reification map $\Reify{A\Rightarrow\NS B}$, we are given
          a function $f : \EL\NS{A}\to \EL\NS{B}$ and must exhibit a normal form
          $\Reify{A\Rightarrow\NS B}f : \Ext{\NfTm\,\prn{A\Rightarrow B}}{\lambda x.
              f x}$. Applying the normal abstraction constructor $\Con{nfLam}$, we are
          given a variable $x:\Var\,\bags{A}$ and must construct a normal form in
          $\Ext{\NfTm\,\bags{B}}{f x}$. Applying reification at $B$, it suffices to
          give an element of $\Ext{\EL\NS{B}}{f x}$; applying $f$ itself, we need an
          element of $\Ext{\EL\NS{A}}{x}$ which we obtain from reflection at $A$ and
          the neutral variable constructor $\Con{neVar}$. To summarize:
          \[
            \Reify{A\Rightarrow\NS B}f =
            \Con{neLam}\,\ReifyTp{A}\,\ReifyTp{B}\,\prn{
              \lambda x.
              \Reify{B}f\,\prn{
                \Reflect{A}
                \Con{neVar}\,\ReifyTp{A}\,x
              }
            }
          \]
  \end{enumerate}

  We leave the construction of $\lambda$-abstraction and application to the
  reader, as they are automatic by the fact that $\EL\NS A \to \EL\NS B$ is
  itself a function space.
\end{construction}

\begin{construction}[The base type in normalization spaces]
  For the base type, we must construct the following three constants:
  \begin{align*}
    \Con{O}\NS   & : \Ext{\TP\NS}{\Con{O}}             \\
    \Con{yes}\NS & : \Ext{\EL\NS\Con{O}\NS}{\Con{yes}} \\
    \Con{no}\NS  & : \Ext{\EL\NS\Con{O}\NS}{\Con{no}}
  \end{align*}

  \begin{enumerate}
    \item We choose $\ReifyTp{\Con{O}\NS}$ to be $\Con{nfO}$.

    \item We will let $\EL\NS\Con{O}\NS$ be the type $\NfTm\,\Con{O}$ of normal
          forms in the base type itself.

    \item The reflection map is given by $\Con{nfNeO} : \Prod{x :
              \NeTm\,\Con{O}}\Ext{\NfTm\,\Con{O}}{x}$.

    \item The reification map given by the identity function.

  \end{enumerate}

  Because we have chosen $\EL\NS\Con{O}\NS = \NfTm\,\Con{O}$, we may interpret
  $\Con{yes}\NS,\Con{no}\NS$ as $\Con{nfYes},\Con{nfNo}$ respectively.
\end{construction}

\subsection{From normalization spaces to a natural model of type theory}\label{sec:externalizing-normalization-spaces}

The results of \cref{sec:stc} culminated with a topos $\GlTop$ of computability
spaces $\GlTop$ equipped with a universe $\NN$ of \DefEmph{normalization
  spaces}, closed under the connectives of our type theory in a way that
restricts under the open immersion $\Mor|open
  immersion|[\bj]{\ClTop{\Mod}}{\GlTop}$ to the corresponding constructs of the
natural model $\Mod$. In this section, we aim to use those constructions as the
basis for an actual natural model $\Modd$ over $\Mod$; later we will
instantiate these results with $\Mod$ taken to be the bi-initial model $\IMod$.
In particular, we shall apply the results of
\cref{sec:universe-to-natural-model} to transform the universe $\NN$ of
normalization spaces into a genuine natural model $\Modd = \Extern[\GCat]{\NN}$ where
$\GCat\subseteq \Sh{\GlTop}$ is some suitable full subcategory of ``test
objects'' containing all $\NN$-contextual objects. In order to choose a
suitable subcategory $\GCat$, we make an auxiliary definition.

\begin{definition}[Atomic computability spaces]\label{def:atomic-space}
  An object of $\Sh{\GlTop}$ is called an \DefEmph{atomic computability space}
  when it lies in the image of the embedding
  $\EmbMor[\pprn{-}]{\CX[\ModR]}{\Sh{\GlTop}}$ defined as the composite
  $\CX[\ModR]\xrightarrow{\Yo[\CX[\ModR]]}\JDG[\ModR]\xrightarrow{\bi_!}{\Sh{\GlTop}}$.
\end{definition}

We then follow \citet{uemura:2022:coh} in choosing $\GCat$ be the smallest
$\NN$-contextual full subcategory of $\Sh{\GlTop}$ containing all atomic
computability spaces (so we may write $\EmbMor[\pprn{-}]{\CX[\ModR]}{\GCat}$).
We will write $\EmbMor[I\Sub{\GCat}]{\GCat}{\Sh{\GlTop}}$ for the full
subcategory inclusion, and $\Mor[N\Sub{\GCat}]{\Sh{\GlTop}}{\Psh{\GCat}}$ for
the corresponding nerve functor that sends each computability space
$X\in\Sh{\GlTop}$ to its \DefEmph{functor of $\GCat$-valued points}.

\begin{definition}
  We define the \DefEmph{normalization model} $\Modd$ to be the externalization
  $\Extern[\GCat]{\NN}$ of $\NN$ at the smallest $\NN$-contextual full
  subcategory $\GCat\subseteq\Sh{\GlTop}$ containing all atomic computability
  spaces.
\end{definition}

\begin{lemma}
  We have a morphism of natural models $\Mor[P]{\Modd}{\Mod}$
  preserving all type structure (function spaces and the base type).
\end{lemma}

\begin{proof}
  The underlying functor $\Mor[P]{\GCat}{\CX[\Mod]}$ can be defined to factor
  like so:
  \[
    \DiagramSquare{
      nw = \GCat,
      ne = \Sh{\GlTop},
      sw = \CX[\Mod],
      se = \JDG[\Mod],
      south = \Yo[\CX[\Mod]],
      east = \bj^*,
      west = P,
      north/style = embedding,
      south/style = embedding,
      west/style = {->,exists},
      north = I\Sub{\GCat},
    }
  \]

  That $\bj^*I\Sub{\GCat}$ factors through the Yoneda embedding follows from
  \cref{lem:display-of-variables-square} and the fact that the property of
  restricting along $\bj$ to a representable is $\NN$-contextual in the sense
  of \cref{def:contextual-class}. We omit the rest of the construction of the
  morphism because it is routine and uninteresting.
\end{proof}

\subsection{The normalization result}\label{sec:normalization-result}

Now instantiate the constructions before by setting $\Mod$ to be the bi-initial
$\IMod$ natural model closed under the specified connectives, and suppose that
$\Mor[\rho]{\ModR}{\IMod}$ is the bi-initial model of variables over $\IMod$.
By the universal property of $\IMod$, we have a section
$\Mor|>->|[S]{\IMod}{\Modd}$ to the projection
$\Mor|->>|[P]{\Modd}{\IMod}$ that we constructed in
\cref{sec:externalizing-normalization-spaces}. The underlying functor of this
section $\Mor[S]{\CX[\IMod]}{\GCat}$ sends each context from $\IMod$ to its
glued interpretation; with this in hand, we make the following definition by
analogy with \cref{def:atomic-space}.

\begin{definition}[Canonical computability spaces]\label{def:canonical-space}
  An object of $\Sh{\GlTop}$ is called a \DefEmph{canonical computability space} when it lies
  in the image of the functor
  $\Mor[\bbrk{-}]{\CX[\ModR]}{\GCat\subseteq\Sh{\GlTop}}$ defined as the
  composite
  $\CX[\ModR]\xrightarrow{\rho}\CX[\IMod]\xrightarrow{S}\GCat$.
\end{definition}

We conclude with some observations that relate $\pprn{-}$ and $\bbrk{-}$ to
the internal language of $\Sh{\GlTop}$.

\begin{construction}[Internalizing types from the model of variables]\label{con:internalizing-types-from-model-of-variables}
  Morphisms $\Mor{\Yo{\Gamma}}{\TP[\ModR]}$ in $\JDG[\ModR]$ can be canonically identified with
  morphisms $\Mor{\pprn{\Gamma}}{\TP}$ in $\Sh{\GlTop}$ by means of the
  following composite natural isomorphism:
  \begin{align*}
     & \Hom{\JDG[\ModR]}{\Yo[\CX[\ModR]]{\Gamma}}{\TP[\ModR]}
    \\
     & \quad\cong
    \Hom{\JDG[\ModR]}{\Yo[\CX[\ModR]]{\Gamma}}{\brho^*\TP[\IMod]}
    \\
     & \quad\cong
    \Hom{\JDG[\IMod]}{\Yo[\CX[\IMod]]{\rho\Gamma}}{\TP[\IMod]}
    \\
     & \quad\cong
    \Hom{\Sh{\GlTop}}{\bj_*\Yo[\CX[\IMod]]{\rho\Gamma}}{\TP}
    \\
     & \quad\cong
    \Hom{\Sh{\GlTop}}{\OpMod\pprn{\Gamma}}{\TP}
    \\
     & \quad\cong
    \Hom{\Sh{\GlTop}}{\pprn{\Gamma}}{\TP}
  \end{align*}

  We shall write $\pprn{-},\pprn{-}\Sup{-1}$ for the forward and inverse
  directions of the natural isomorphism above.
\end{construction}

\begin{observation}\label{obs:atomic-space-context-extension}
  Let $\Mor[A]{\Yo{\Gamma}}{\TP[\ModR]}$ be a type in $\ModR$; then the atomic
  computability space $\pprn{\Gamma.A}$ is canonically isomorphic to the
  dependent sum
  $\Sum{\gamma:\pprn{\Gamma}}\Con{Var}\prn{\pprn{A}\Sup{-1}\gamma}$.
\end{observation}

\begin{lemma}
  The projection map $\Mor[\pi_1]{\Sum{A:\TP[\IMod]}\Con{Var}\,A}{\TP[\IMod]}$ is relatively
  representable by an atomic computability space.
\end{lemma}

\begin{proof}
  Let $\pprn{\Gamma}$ an atomic computability space and let
  $\Mor[B]{\pprn{\Gamma}}{\TP[\IMod]}$; we compute the fiber of $\pi_1$ as
  follows:
  \[
    \DiagramSquare{
      nw = \Sum{\gamma:\pprn{\Gamma}}\Con{Var}\,\prn{B\gamma},
      sw = \pprn{\Gamma},
      ne = \Sum{A:\TP[\IMod]}{\Con{Var}\,A},
      se = \TP[\IMod],
      east = \pi_1,
      south = B,
      north = {\prn{\gamma,x}\mapsto \prn{B\gamma,x}},
      nw/style = pullback,
      width = 5cm,
    }
  \]

  By \cref{obs:atomic-space-context-extension}, the pullback above is isomorphic
  to the projection
  $\Mor[\pprn{p\Sub{\pprn{B}\Sup{-1}}}]{\pprn{\Gamma.\pprn{B}\Sup{-1}}}{\pprn{\Gamma}}$.
\end{proof}

\begin{observation}\label{obs:canonical-space-context-extension}
  Let $\Mor[A]{\Yo{\Gamma}}{\TP[\ModR]}$ be a type in $\ModR$; recalling that
  $\bbrk{-}=S\circ\rho$ tracks a morphism of natural models, we have a type
  $\Mor[\prn{S\circ\rho}\Sub{\TP}\cdot A]{\Yo{\bbrk{\Gamma}}}{\TP[\Extern[\GCat]{\NN}]}$ in
  $\Modd$, which can equally well be described as a map
  $\Mor[\bbrk{A}]{\bbrk{\Gamma}}{\TP\NS}$ in $\Sh{\GlTop}$. From this
  perspective, the canonical computability space $\bbrk{\Gamma.A}$ is the dependent sum
  $\Sum{\gamma:\bbrk{\Gamma}}\EL\NS{\bbrk{A}\gamma}$.
\end{observation}

\NewDocumentCommand\Hydrate{}{\mathord{\nearrow}}
\NewDocumentCommand\ModH{}{\Kwd{H}}

\subsubsection{The functors of atomic and canonical points}

Given a space $X\in \Sh{\GlTop}$, a \DefEmph{atomic point} of $X$ is defined to
be a generalized element of $X$ defined on an atomic computability space $\pprn{\Gamma}$;
likewise, a \DefEmph{canonical point} of $X$ is defined to be a generalized
element of $X$ defined on a canonical computability space $\bbrk{\Gamma}$. Thus the
\DefEmph{functors of (atomic, canonical) points} of $X$ are the presheaves
$\Hom{\Sh{\GlTop}}{\pprn{-}}{X}$, $\Hom{\Sh{\GlTop}}{\bbrk{-}}{X}$ respectively in
$\JDG[\ModR]$.

\begin{definition}[Restricting to a functor of points]
  Let $\Mor[F]{\CX[\ModR]}{\Sh{\GlTop}}$ be a functor such that
  $\bj^*\circ F \cong \Yo[\CX[\IMod]]\circ\rho$; for any $X\in\Sh{\GlTop}$, the
  \DefEmph{functor of $F$-valued points} of $X$ is defined to be the presheaf
  $\Hom{\Sh{\GlTop}}{F-}{X}$ in $\JDG[\ModR]$.  We define the
  \DefEmph{restriction of $X$ to its functor of $F$-valued points} to be the
  space $X_F\in\Sh{\GlTop}$ determined by the following natural transformation
  $\Mor{\Hom{\Sh{\GlTop}}{F-}{X}}{\brho^*\bj^*X}$:
  \[
    \Hom{\Sh{\GlTop}}{F-}{X}
    \xrightarrow{}
    \Hom{\JDG[\IMod]}{\bj^*F-}{\bj^*X}
    \xrightarrow{\cong}
    \Hom{\JDG[\IMod]}{\Yo{\rho-}}{\bj^*X}
    \xrightarrow{}
    \brho^*\bj^*X
  \]

  Given a natural transformation $\Mor[\alpha]{F}{G}$ between two such functors,
  the precomposition map $\Hom{\Sh{\GlTop}}{\alpha-}{X}$ induces a vertical
  reindexing map $\Mor[X\Sub{\alpha}]{X\Sub{G}}{X\Sub{F}}$.
\end{definition}

\begin{lemma}
  For any space $X\in\Sh{\GlTop}$, the functor of canonical points
  $\Hom{\Sh{\GlTop}}{\pprn{-}}{X}$ is canonically isomorphic to the restriction
  $\bi^*X$ of $X$ along the closed immersion
  $\Mor|closed immersion|[\bi]{\ClTop{\ModR}}{\GlTop}$.
\end{lemma}

\begin{proof}
  This follows by adjointness and the definition $\pprn{-}=\bi_!\circ\Yo$.
\end{proof}

\begin{corollary}
  The restriction $X\Sub{\pprn{-}}$ of any space $X\in\Sh{\GlTop}$ to its
  functor of atomic points is canonically isomorphic to $X$ itself.
\end{corollary}

\begin{construction}[{Internalizing the action of $S$ on types}]\label{con:internalize-eval-tp}
  The map $\Mor|>->|[S]{\IMod}{\Modd}$ determined by the universal property of
  the bi-initial model carries an action $S\Sub{\TP}\cdot- $ that transforms a
  type $\Mor[A]{\Yo{\Gamma}}{\TP[\IMod]}$ in the bi-initial model to a type
  $\Mor[S\Sub{\TP}\cdot A]{\Yo{S\Gamma}}{\TP[\Modd]}$ in the normalization
  model. This map internalizes directly into $\Sh{\GlTop}$ as a \emph{vertical
    map} from $\Mor{\TP}{\TP\NS\Sub{\bbrk{-}}}$ from $\TP$ to the restriction
  $\TP\NS\Sub{\bbrk{-}}$ of $\TP\NS$ to its functor of canonical points.
  To define a vertical map $\Mor{\TP}{\TP\NS\Sub{\bbrk{-}}}$ is
  the same as to define a section of the projection map
  $\Mor{\Hom{\Sh{\GlTop}}{\bbrk{-}}{\TP\NS}}{\brho^*\TP[\IMod]}$:
  \[
    \brho^*\TP[\IMod] \xrightarrow{\cong}
    \Hom{\JDG[\IMod]}{\Yo{\rho-}}{\TP[\IMod]}
    \xrightarrow{S\Sub{\TP}\cdot-}
    \Hom{\Psh{\GCat}}{\Yo{\bbrk{-}}}{N\Sub{\GCat}\TP\NS}
    \xrightarrow{\cong}
    \Hom{\Sh{\GlTop}}{\bbrk{-}}{\TP\NS}
  \]
\end{construction}

\begin{remark}[Toward an internal evaluation map]\label{rem:toward-internal-eval}
  The vertical map $\Mor{\TP}{\TP\NS\Sub{\bbrk{-}}}$ internalizing the action
  of $S$ on types from \cref{con:internalize-eval-tp} is a good first step,
  what we \emph{need} for our results is an unrestricted vertical map
  $\Mor{\TP}{\TP\NS}$. We will do so by exhibiting for \emph{any} $X$ a
  canonical vertical map $\Mor{X\Sub{\bbrk{-}}}{X}$; recalling that $X\cong
    X\Sub{\pprn{-}}$, it evidently suffices to define a (suitably vertical)
  natural transformation $\Mor{\pprn{-}}{\bbrk{-}}$ from the functors of atomic
  points to the functors of canonical points, which we shall refer in
  \cref{sec:inserter} as \DefEmph{hydration}.
\end{remark}

\subsubsection{Hydration of variables via Bocquet, Kaposi, and Sattler's \emph{inserter}}\label{sec:inserter}

The goal of this section is to define a suitably vertical natural
transformation $\Mor{\pprn{-}}{\bbrk{-}}$ that ``hydrates'' an element of an
atomic computability space into an element of the corresponding canonical
computability. Reindexing along this natural transformation, we would then
obtain a map $\Mor{X\Sub{\bbrk{-}}}{X\Sub{\pprn{-}}\cong X}$ that we could use
to define an internal evaluation map $\Mor{\TP}{\TP\NS}$ as in
\cref{rem:toward-internal-eval}.

We shall view both $\CX[\ModR]$ and $\GCat$ as categories \emph{displayed} over
$\JDG[\IMod]$ via the functors
$\Mor[\Yo[\CX[\IMod]]\circ \rho]{\CX[\ModR]}{\JDG[\IMod]}$ and
$\Mor[\bj^*]{\GCat}{\JDG[\IMod]}$. We observe that both $\pprn{-}$ and $\bbrk{-}$
lift into the slice $\Sl{\CAT}{\JDG[\IMod]}$, as witnessed by the following
diagram:
\[
  \begin{tikzpicture}[diagram]
    \node (nc) {$\CX[\ModR]$};
    \node[left = of nc] (nw) {$\GCat$};
    \node[right = of nc] (ne) {$\GCat$};
    \node[below = of nc] (sc) {$\JDG[\IMod]$};
    \draw[->] (nc) to node[above] {$\pprn{-}$} (nw);
    \draw[->] (nc) to node[above] {$\bbrk{-}$} (ne);
    \draw[->] (nc) to node[upright desc] {$\Yo[\CX[\IMod]]\circ\rho$} (sc);
    \draw[->] (ne) to node[sloped,below] {$\bj^*$} (sc);
    \draw[->] (nw) to node[sloped,below] {$\bj^*$} (sc);
  \end{tikzpicture}
\]

Stated now with more precision, our goal is then define a 2-cell
$\Mor[\Hydrate]{\pprn{-}}{\bbrk{-}}$ in the slice $\Sl{\CAT}{\JDG[\IMod]}$ that
``hydrates'' an element of an atomic computability space to an element of the
corresponding canonical computability space. Our construction follows that of
\citet[Appendix A]{uemura:2022:coh}, which is itself modeled on the original
more cryptic formulation by \citet{bocquet-kaposi-sattler:2021}. In particular,
we shall define a model of variables $\ModH$ over $\ModR$ from which we can
extract the desired hydration map.  This is essentially an inductive argument
that will be carried out using the universal property of $\ModR$ as the
bi-initial model of variables over $\IMod$.

\begin{construction}[The hydration model]
  We choose $\CX[\ModH]$ to be the \DefEmph{inserter object} determined by the
  morphisms $\pprn{-},\bbrk{-}$ in $\Sl{\CAT}{\JDG[\IMod]}$. An object of the
  inserter $\CX[\ModH]$ is a pair of an object $\Gamma\in\CX[\ModR]$ and a
  \emph{vertical} map $\Mor[\eta_\Gamma]{\pprn{\Gamma}}{\bbrk{\Gamma}}$; a morphism
  from $\prn{\Delta,\eta_\Delta}$ to $\prn{\Gamma,\eta_\Gamma}$ is given by a morphism
  $\Mor[\gamma]{\Delta}{\Gamma}$ such that the following square commutes:
  \[
    \DiagramSquare{
      nw = \pprn{\Delta},
      sw = \bbrk{\Delta},
      ne = \pprn{\Gamma},
      se = \bbrk{\Gamma},
      north = \pprn{\gamma},
      south = \bbrk{\gamma},
      west = \eta_\Gamma,
      east = \eta_\Delta,
    }
  \]

  There is an evident projection functor
  $\Mor[H]{\CX[\ModH]}{\CX[\ModR]}$ sending each
  $\prn{\Gamma,\eta_\Gamma}$ to $\Gamma$. We define $\TP[\ModH]\in\JDG[\ModH]$ to
  be the presheaf $H^*\TP[\ModR]$; in order to define
  $\EL[\ModH]\in\Sl*{\JDG[\ModH]}{\TP[\ModH]}$, we first describe the
  comprehension of a given element
  $A\in \TP[\ModH]\prn{\Gamma,\eta_\Gamma}$ as an object
  $\prn{\Gamma.A, \eta\Sub{\Gamma.A}}\in \Sl*{\CX[\ModH]}{\prn{\Gamma,\eta_\Gamma}}$.
  In particular, let $\Gamma.A$ be the corresponding comprehension in $\ModR$ as
  below:
  \[
    \DiagramSquare{
      nw = \Yo\prn{\Gamma.A},
      sw = \Yo{\Gamma},
      west = \Yo{p_A},
      south = A,
      east = \Proj[\ModR],
      ne = \EL[\ModR],
      se = \TP[\ModR],
      nw/style = pullback,
    }
  \]

  We will define a vertical map
  $\Mor[\eta\Sub{\Gamma.A}]{\pprn{\Gamma.A}}{\bbrk{\Gamma.A}}$ fitting into the
  following commuting square:
  \[
    \DiagramSquare{
      north/style = {exists,->},
      nw = \pprn{\Gamma.A},
      ne = \bbrk{\Gamma.A},
      sw = \pprn{\Gamma},
      se = \bbrk{\Gamma},
      south = \eta_\Gamma,
      north = \eta\Sub{\Gamma.A},
      west = \pprn{p_A},
      east = \bbrk{p_A},
    }
  \]

  Using
  \cref{obs:atomic-space-context-extension,obs:canonical-space-context-extension},
  we see that such a map can be defined using the following \emph{internal} variable
  hydration map defined using the reflection map of any normalization space:

  \iblock{
    \mrow{
      \Con{hydrate} :
      \Prod{A:\TP\NS}
      \Con{Var}\,\bags{A}
      \tovrt
      \EL\NS{A}
    }
    \mrow{
      \Con{hydrate}\,A\,x =
      \Reflect{A}{
        \Con{neVar}\,\prn{\ReifyTp{A}}\,x
      }
    }
  }

  The projection functor $\Mor[H]{\CX[\ModH]}{\CX[\ModR]}$ can now be seen to
  track a morphism of natural models $\Mor[H]{\ModH}{\ModR}$; moreover, this
  morphism exhibits $\ModH$ by definition as a \emph{model of variables} over
  $\ModR$.
\end{construction}

\begin{construction}[The hydration map]
  As $\Mor[H]{\ModH}{\ModR}$ is a model of variables over $\ModR$, the composite
  $\Mor[\rho\circ H]{\ModH}{\IMod}$ is also a model of variables over $\IMod$. As $\ModR$
  is assumed to be the bi-initial model of variables, we have an essentially
  unique section $\Mor{\ModR}{\ModH}$ over $\IMod$. The underlying functor of
  this section sends each context $\Gamma\in\CX[\ModR]$ to a morphism
  $\Mor[\eta_\Gamma]{\pprn{\Gamma}}{\bbrk{\Gamma}}$, and functoriality
  guarantees that this assignment is natural. Therefore, we may define
  $\Mor[\Hydrate]{\pprn{-}}{\bbrk{-}}$ componentwise by
  $\Hydrate_\Gamma=\eta_\Gamma$.
\end{construction}

\subsubsection{The normalization map and its injectivity}

By reindexing along our vertical hydration map
$\Mor[\Hydrate]{\pprn{-}}{\bbrk{-}}$, we therefore obtain a vertical map
$\Mor[X\Sub{\Hydrate}]{X\Sub{\bbrk{-}}}{X\Sub{\pprn{-}}\cong X}$. As we see
below, this is enough to fulfill the problem posed by
\cref{rem:toward-internal-eval}.

\begin{construction}[The internal evaluation map]\label{con:eval-map}
  We shall now exhibit a vertical evaluation map
  $\Mor[\bbrk{-}\Sub{\TP}]{\TP}{\TP\NS}$ within $\Sh{\GlTop}$ sending
  each type $\IMod$-type to the normalization space chosen by our model.
  \[
    \TP\xrightarrow{\text{\cref{con:internalize-eval-tp}}} \TP\NS\Sub{\bbrk{-}} \xrightarrow{\TP\NS\Sub{\Hydrate}} \TP\NS
  \]
\end{construction}

\begin{construction}[The internal normalization map]\label{con:norm-map}
  We may compose the internal evaluation map
  $\Mor[\bbrk{-}\Sub{\TP}]{\TP}{\TP\NS}$ with the vertical projection
  $\Mor[\ReifyTp{-}]{\TP\NS}{\NfTp}$ of normal forms from normalization spaces
  to obtain a vertical \DefEmph{normalization map}
  $\Mor[\Con{norm}\Sub{\TP}]{\TP}{\NfTp}$ that takes any element of
  $\TP$ to its normal form.
\end{construction}

\begin{observation}\label{obs:norm-injective}
  The internal normalization map $\Mor[\Con{norm}\Sub{\TP}]{\TP}{\NfTp}$
  is a monomorphism, as it is a section of the unit map
  $\Mor{\NfTp}{\OpMod{\NfTp}\cong\TP}$.
\end{observation}

\subsection{Injectivity of type constructors}

In \cref{question:injectivity}, we have asked whether
$\Mor[\Rightarrow\Sub{\IMod}]{\TP[\IMod]\times\TP[\IMod]}{\TP[\IMod]}$ is a
monomorphism in $\JDG[\IMod]$. We can now answer in the affirmative, by virtue
of the normalization result (\cref{sec:normalization-result}).

\begin{lemma}\label{lem:fringe-functor-is-surjection}
  The morphism of topoi $\Mor[\brho]{\ClTop{\ModR}}{\ClTop{\IMod}}$ is a geometric
  surjection.
\end{lemma}

\begin{proof}
  This follows from \cref{lem:geometric-surjection}, since the bi-initial model
  is always democratic~\citep{uemura:2021:thesis}.
\end{proof}

\begin{theorem}[Injectivity of type constructors]\label{thm:main-theorem}
  The function space constructor
  $\Mor[\prn{\Rightarrow\Sub{\IMod}}]{\TP[\IMod]\times\TP[\IMod]}{\TP[\IMod]}$ is a
  monomorphism in $\JDG[\IMod]$.
\end{theorem}

\begin{proof}
  As $\Mor[\brho]{\ClTop{\ModR}}{\ClTop{\IMod}}$ is a surjection
  (\cref{lem:fringe-functor-is-surjection}), its inverse image functor is
  faithful; as $\Mor|closed immersion|[\bi]{\ClTop{\ModR}}{\GlTop}$ is an
  embedding, its direct image is (fully) faithful. As faithful functors reflect
  monomorphisms, it suffices for us to show that $\bi_*\brho^*\prn{\Rightarrow\Sub{\IMod}}$ is
  a monomorphism in $\JDG[\ModR]$. Since $\brho^* \cong \bi^*\bj_*$ we have
  $\bi_*\brho^*\prn{\Rightarrow\Sub{\IMod}} \cong \bi_*\bi^*\bj_*\prn{\Rightarrow\Sub{\IMod}} =
    \ClMod\bj_*\prn{\Rightarrow\Sub{\IMod}}$. Hence it is enough to show that
  $\ClMod\bj_*\prn{\Rightarrow\Sub{\IMod}}$ is a monomorphism.

  Switching to the internal language, we suppress the embedding $\bj_*$ and aim
  to check that the function $\ClMod\prn{\Rightarrow} :
    \ClMod\prn{\TP\times\TP}\to \ClMod\TP$ is injective. Fixing $A,A',B,B':\TP$
  such that $\prn{A\Rightarrow B} = \prn{A'\Rightarrow B'}$, our goal is to
  check that $\ClMod\prn{\prn{A,B} = \prn{A',B'}}$. We know that
  $\Con{norm}\Sub{\TP}\prn{A\Rightarrow B} =
    \Con{norm}\Sub{\TP}\prn{A'\Rightarrow B'}$; unfolding the definition of
  $\Con{norm}\Sub{\TP}$ induced by the normalization model in
  \cref{con:eval-map,con:norm-map} we conclude that
  $\Con{nfFun}\prn{\Con{norm}\Sub{\TP}A,\Con{norm}\Sub{\TP}B} =
    \Con{nfFun}\prn{\Con{norm}\Sub{\TP}A',\Con{norm}\Sub{\TP}B'}$. Our goal then
  follows from the modal injectivity of normal form constructors
  (\cref{lem:modal-injectivity}) together with our \cref{obs:norm-injective}
  that the normalization function is injective.
\end{proof}

\section{Concluding remarks}

We have at long last shown in \cref{thm:main-theorem} how to prove that the
type constructor for function spaces is a monomorphism in the bi-initial model
of type theory with function spaces on a base type with two constants. A few
things deserve additional comment.

\paragraph{Extension to more sophisticated results}

We have focused on the injectivity of ordinary function spaces for the sake of
simplicity, but the methods exposed herein also apply to dependent product,
dependent sums, \etc. Likewise, our methods extend readily to prove more
difficult results, including the fact that the normalization function is not
only a section but in fact an isomorphism. From these results, one may deduce a
solution to the word problem for Martin-L\"of type theory. Finally, these
methods can be adapted to apply to much more sophisticated type theories,
including cubical type
theory~\citep{sterling:2021:thesis,sterling-angiuli:2021}, multi-modal type
theory~\citep{gratzer:2022:lics}, and even ``$\infty$-type
theories''~\citep{uemura:2022:coh}.

\paragraph{Emphasis of universal properties over explicit constructions}

At every stage in our development, we have worked as much as possible with
invariant universal properties rather than explicit constructions. For
instance, we worked with the (2,1)-categorical universal property of bi-initial
natural model not because we do not think that the concrete syntax of type
theory is important, but because we want our proofs to be flexible enough to
apply to \emph{any} correct implementation of this concrete syntax, \ie any
presentation that can be shown to satisfy the universal property. The concrete
presentation of type theoretic syntax is both highly non-trivial and deeply
obscure: for this reason, it cannot be counted as a virtue for a proof to be
applicable only to a specific obscure presentation that is likely to be
superseded as the winds of fashion blow one way or another.

Likewise, it is possible to give an explicit construction of the ``model of
variables'' in terms of syntactically defined telescopes (see
\citet[\S5.5]{sterling:2021:thesis} for such a construction), but we have
followed the more modular proof technique of
\citet{bocquet-kaposi-sattler:2021} not because we wish to worship abstraction
for abstraction's sake, but because the proof applies to \emph{any}
presentation of the bi-initial model of variables. The flexibility to choose
different presentations is very important for implementation because such
choices can have a significant impact on efficiency; therefore, a modern proof
that is invariant in this way is arguably much closer to practical applications
than the more old-fashioned ones that emphasized explicit constructions.
 
\clearpage
\bibliographystyle{plainnat}
\bibliography{references/refs-bibtex}

\nocite{sga:4}
\nocite{johnstone:topos:1977,johnstone:2002}
\nocite{streicher:2021:fib}
\nocite{coquand:2019}

\nocite{sterling:2021:thesis}
\nocite{awodey:2022:universes}

\end{document}